\documentclass[12pt,a4paper]{article}
\usepackage[truedimen,margin=30mm]{geometry} 

\usepackage{mathrsfs}
\usepackage{amssymb}
\usepackage{amsmath}
\usepackage{ascmac}
\usepackage{amsthm}
\usepackage[dvipdfmx]{graphicx}

\usepackage{natbib}
\usepackage{setspace}
\usepackage{color}
\usepackage{url}

\usepackage{comment}
\usepackage{bm}
\usepackage{times}
\usepackage{titlesec}
\titleformat*{\section}{\large\bfseries}
\titleformat*{\subsection}{\it}

\newtheorem{thm}{Theorem}
\newtheorem{lem}{Lemma}

\newtheorem{prp}{Proposition}

\newtheorem{algo}{Algorithm}

\newtheorem{remark}{Remark}
\def\beh{{\widehat{\beta}}}
\def\Sih{{\widehat{\Sigma}}}

\def\At{{\widetilde{A}}}
\def\Bt{{\widetilde{B}}}

\def\lah{{\widehat \la}}

\def\ep{{\varepsilon}}
\def\la{{\lambda}}
\def\si{{\sigma}}
\def\th{{\theta}}

\def\Si{{\Sigma}}
\def\Ga{\Gamma}

\def\al{{\alpha}}
\def\be{{\beta}}
\def\ga{{\gamma}}
\def\de{{\delta}}
\def\ep{{\varepsilon}}
\def\la{{\lambda}}
\def\si{{\sigma}}
\def\om{{\omega}}
\def\th{{\theta}}

\def\Th{{\Theta}}
\def\De{{\Delta}}
\def\Si{{\Sigma}}
\def\Ga{{\Gamma}}
\def\Om{{\Omega}}
\def\Up{{\Upsilon}}
\def\La{{\Lambda}}
\def\non{{\nonumber}}

\def\Lc{{\cal L}}

\def\Kc{{\cal K}}
\title{{\bf Robust Bayesian Modeling of Counts with Zero inflation and Outliers: Theoretical Robustness and Efficient Computation}}

\date{}
\author{}

\begin{document}

\maketitle
\doublespacing

{
\renewcommand{\thefootnote}{\fnsymbol{footnote}}
\vspace{-1.5cm}
\begin{center}
Yasuyuki Hamura$^1$, Kaoru Irie$^2$, and Shonosuke Sugasawa$^3$
\end{center}
}

\noindent
$^1$Graduate School of Economics, Kyoto University\\
$^2$Faculty of Economics, The University of Tokyo \\
$^3$Faculty of Economics, Keio University

\vspace{5mm}
\begin{center}
{\bf \large Abstract}
\end{center}
Count data with zero inflation and large outliers are ubiquitous in many scientific applications. However, posterior analysis under a standard statistical model, such as Poisson or negative binomial distribution, is sensitive to such contamination. 
This study introduces a novel framework for Bayesian modeling of counts that is robust to both zero inflation and large outliers. In doing so, we introduce rescaled beta distribution and adopt it to absorb undesirable effects from zero and outlying counts.
The proposed approach has two appealing features: the efficiency of the posterior computation via a custom Gibbs sampling algorithm and a theoretically guaranteed posterior robustness, where extreme outliers are automatically removed from the posterior distribution. 
We demonstrate the usefulness of the proposed method by applying it to trend filtering and spatial modeling using predictive Gaussian processes.


\newpage 
%------------------------------------------------------------
%    Introduction
%------------------------------------------------------------
\section{Introduction}

Zero inflation and large outliers are ubiquitous problems in statistical modeling of counts; in many applications the use of a Poisson distribution will result in biased inference and predictive analysis. 
While there are standard solutions for dealing zero inflation such as hurdle \citep{arulampalam1997gets,mullahy1986specification} and zero-inflated models \citep{lambert1992zero,shankar1997modeling}, effective and robust methods against large outlying counts remain limited.
For moderate outliers, standard over-dispersion models (e.g., negative binomial or quasi-Poisson models) can be used. 
However, as demonstrated by ``hot spots'' in crime or epidemiological statistics, the data used in recent applied statistics can contain extremely large outliers. Even if 
such extreme values account for only a small portion of the whole sample, the over-dispersion models are not able to produce robust inferences of structures of interest (e.g., mean regression structure).
Moreover, the two problems---large outliers and zero-inflation---have been discussed separately in the literature. To the best of our knowledge, no easy solution for the robust modeling of counts with both zero inflation and large outliers has yet been found.

Beyond Poisson or negative binomial distributions, flexible parametric families of distributions for count data have been proposed.
Examples include the Conway-Maxwell-Poisson distribution \citep{sellers2010flexible,chanialidis2018efficient}, Poisson inverse Gaussian distribution \citep{barreto2016general}, and hyper-Poisson distribution \citep{saez2013hyper}.
However, these models cannot simultaneously account for zero inflation and large outliers. 
\cite{datta2016bayesian} proposed a heavy-tailed distribution for an intensity parameter of the Poisson rate that can deal with zero inflation and large outlying counts. This flexibility comes at the cost of using special functions in the density function, which makes the posterior computation difficult under complicated hierarchical regression models. 
\cite{hamura2019global} proposed a heavy-tailed distribution that can handle outlying counts, but it does not account for zero inflation.

The general methodological framework for robust Bayesian inference has been studied extensively in recent years; however, the focus of this research has not necessarily been on count data analysis. 
For example, several analyses based on a general posterior distribution \citep{bissiri2016general} replace the standard likelihood function with robust divergence and define a robust posterior distribution \citep[e.g.,] []{jewson2018principles,knoblauch2018doubly,hashimoto2020robust}.
However, robust divergence for count distributions, such as Poisson or negative binomial distributions, entails serious numerical difficulty.
For example, the density power divergence \citep{basu1998robust} for a Poisson distribution involves an intractable infinite sum, for which the posterior computation becomes infeasible. 
Other attempts at general robust Bayesian inference include $c$-posterior \citep{miller2019robust} and fractional posterior \citep{bhattacharya2019bayesian}. 
These approaches can be applied to a wide range of data and models, but the posterior distribution may have complicated forms, especially under hierarchical models. 
For continuous observations, although the standard remedy is to simply replace the normal distribution with heavy-tailed alternatives \citep[e.g.,][]{west1984outlier,Gag2019,hamura2020log}, such distributions cannot be directly employed in modeling counts because of the discrete support of the count distribution. 
Finally, non-Bayesian (frequentist) approaches to robustness have also been studied in the context of count data modeling \citep[e.g.,][]{KF2019}, but these methods cannot be used directly for hierarchical modeling of counts, nor can they be customized easily.

In this study, we introduce a tailored framework for the robust hierarchical modeling of counts.
A key component of the proposed method is families of distributions on the positive real line, scaled beta (SB) and rescaled beta (RSB) distributions, which generate both excess zeros and large counts.
These distributions can replace the exponential/gamma distributions used in the existing models for counts while being conditionally conjugate. 
For practical implementation, we derive a custom Gibbs sampling algorithm for the models with the SB and RSB distributions, which can be applied to many models, including hierarchical spatial regression models and mixed models. 
From the theoretical viewpoint, the RSB distribution has the super heavy tail and is shown to have the posterior robustness property. 
We first prove that, with a high probability, the unexplained zero count can be ignored in the posterior inference with the SB and RSB distributions. 
Moreover, if the RSB distribution is used, the posterior distribution can successfully eliminate information from outliers as they move further away from the model mean, which guarantees that the proposed approach can provide robust point estimation as well as robust uncertainty quantification. 
A clear advantage of this property is that the robustness is achieved even if we do not explicitly model zero-inflation structures and the mechanism of generating outlying counts, which are not of interest in most applications.

This paper is organized as follows. 
Section~\ref{sec:RSB} introduces the SB and RSB distributions that play a central role in the proposed robust methodology. 
In Section~\ref{sec:method}, we introduce the proposed method, provide an efficient posterior computation algorithm, and prove the theoretical robustness properties. 
In Section~\ref{sec:exm}, we demonstrate the proposed methods using three models: a generalized linear model for counts, locally adaptive smoothing for trend filtering, and spatial models based on Gaussian processes. 
Section~\ref{sec:remark} gives some concluding remarks. 
All technical issues, including mathematical proofs, are presented in the Supplementary Material. 

{\bf Notation:} We use $N(\mu , \Sigma)$, ${\rm Ex}(v)$, ${\rm Ga}(a,b)$, ${\rm IG}(a,b)$, ${\rm Be}(a,b)$ and ${\rm Po}(\lambda)$ for the multivariate normal, exponential (mean $1/v$), gamma (mean $a/b$), inverse gamma, beta, and Poisson distributions, respectively. We also use this notation for the density functions. For example, the density of ${\rm Ex}(v)$ evaluated at $x$ is written as ${\rm Ex}(x|v)$.

%---------------------------------------------------------------------
%    Distribution theory 
%---------------------------------------------------------------------
\section{Scaled and Rescaled beta distributions}\label{sec:RSB}

In this section, we summarize the distributional theory used to model the Poisson rate. First, we review the scaled beta (SB) distribution, also known as prime beta or inverted beta distributions, whose density is given by:  
\begin{equation*}
\pi_{\rm SB}(\eta ;a,b) = \frac{1}{B(a,b)} \eta ^{a-1} (1+\eta )^{-(a+b)}, \ \ \ \ \ \eta >0,
\end{equation*}
where $a$ and $b$ are positive shape parameters, and $B(a,b)$ is the beta function. 
This distribution is obtained by the scale transformation of a beta-distributed random variable. 
An important special case is the half-Cauchy distribution (for $\sqrt{\eta}$) obtained by setting $(a,b)=(1/2, 1/2)$. 

A new class of distributions is derived from the SB distribution above by making the change of variables. For $\eta_{\mathrm{SB}} \sim \mathrm{SB}(a,b)$, we set $\eta_{\mathrm{SB}}=\log (1+\eta)$ to obtain the density of $\eta \in (0,\infty)$ given by 
\begin{equation}\label{RSB-general}
\pi_{\rm RSB}(\eta ;a,b)=\frac1{B(a,b)}\frac{\{\log(1+\eta )\}^{a-1}}{1+\eta }\frac{1}{\{1+\log(1+\eta )\}^{a+b}}, \ \ \ \eta >0, 
\end{equation}
which we call {\it the rescaled beta (RSB) distribution}. 
A special case of this class of distributions, $a=1$, is known as the log-Pareto distribution (e.g., \citealt{cormann2009generalizing}), and its robustness property as a shrinkage prior under count observations is considered in \cite{hamura2019global}. A similar idea that involves introducing log-terms to well-known distributions can also be seen in \cite{schmidt2018log}, but it results in a different class of distributions.

A notable aspect of the RSB distribution in (\ref{RSB-general}) is its super-heavy-tailed property. 
The density function of the RSB distribution in the tail is evaluated as 
\begin{equation*}
\pi_{\rm RSB}(\eta ;a,b) \sim \eta ^{-1} (\log \eta ) ^{- (1+b)} \ \ \ \ \ \mathrm{as} \ \ \ \eta \to \infty ,
\end{equation*}
where $f(\eta ) \sim g(\eta )$ indicates that the limit of $f(\eta )/g(\eta )$ is finite. In other words, the RSB distribution is log-regularly varying \citep[e.g.,][]{desgagne2015robustness}. 
This tail is heavier than those of the SB distributions ($\pi_{\rm SB}(\eta) \sim \eta ^{-(1+b)}$) and the half-Cauchy and half-$t$ distributions. 
By contrast, the SB and RSB densities at the origin are 
\begin{equation} \label{eq:zerodens}
\pi_{\rm SB}(\eta ; a,b) \sim \pi_{\rm RSB}(\eta ;a,b) \sim \eta ^{a-1} \ \ \ \ \ \mathrm{as} \ \ \ \eta \downarrow 0.
\end{equation}
That is, the behaviors of the RSB and SB densities with shape $a$ are the same near the origin. A spike in the density at $\eta =0$ appears if and only if $0<a<1$. In the left panel of Figure~\ref{fig:dens}, some examples of SB and RSB densities are plotted, from which we can visually confirm the behaviors of SB and RSB densities at the origin and tail. The tail probability of the RSB distribution can be computed explicitly by transforming its density back to the beta density. Let $t\sim {\rm Be}(a,b)$ and $F$ be the distribution function of $t$, or $F(t_0) = P[t\le t_0]$ for $t_0 > 0$, which is computed as the scaled incomplete beta function. Then, for $\eta \sim \pi _{\rm RSB}(\eta;a,b)$, we have 
\begin{equation*}
P[\eta \le \eta _0] = F\left( \frac{\log (1+\eta _0)}{1+\log (1+\eta _0)} \right) \ \ \ \ \ \mathrm{for} \ \ \ \eta _0>0.
\end{equation*}
The right panel of Figure~\ref{fig:dens} shows the distribution functions of the SB and RSB distributions. The tail probability of the RSB distributions converges to unity as $u_0\to \infty$, but it does so very slowly, reflecting the characteristics of its log-regular variation. The probability of having values greater than or equal to 1,000 is at least 0.2, allowing extremely large values to be realized. From these observations, we set $a=b=1/2$ as the default for the subsequent analyses in this research, but it can be generalized to arbitrary $0<a<1$ and $b>0$, depending on the prior belief of interest.

\begin{figure}[!htb]
\centering
\includegraphics[width=14cm,clip]{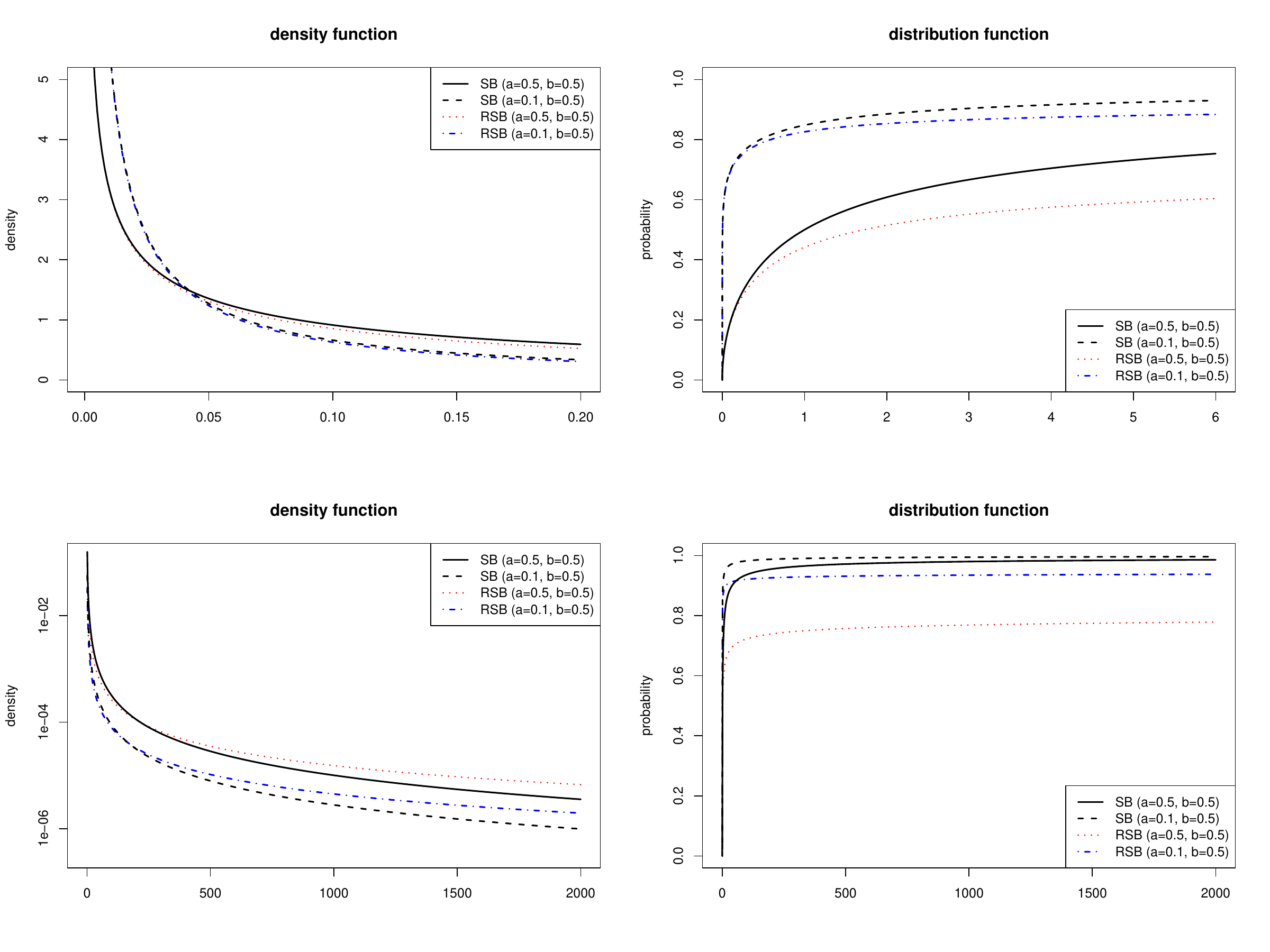}
\caption{Probability densities and distribution functions of the scaled-beta (SB) and rescaled-beta (RSB) distributions. The densities of SB$(0.5,0.5)$, SB$(0.1,0.5)$, RSB$(0.5,0.5)$, and RSB$(0.1,0.5)$ are shown near the origin (top left) and tail (bottom left). The distribution functions of the four distributions are shown near the origin (top right) and tail (bottom right).
\label{fig:dens}
}
\end{figure}

The SB and RSB distributions can be represented as hierarchical models of well-known distributions. 
For the RSB distributions, we consider two representations: one for posterior computation and the other for interpretation.

\begin{thm}\label{thm:int}
The RSB distribution has the following integral expressions:
\begin{equation}\label{RSB-integral}
\pi _{\rm RSB}(\eta ;a,b) 
=\int_{\mathbb{Re}_+^3} {\rm Ex}(\eta |u) \ {\rm Ga}(u|v+w,1) \ g(v,w ) \ dudvdw, 
\end{equation}
where $\Gamma (\cdot)$ is the gamma function, and $g(v,w)$ is the bivariate joint probability density defined by 
\begin{equation}\label{density-g}
g(v,w) = \frac{1}{B(a,b)\Gamma (1-a) \Gamma (a+b)}  \frac{v^{-a}w^{a+b-1}e^{-w} }{v+w}. 
\end{equation}
\end{thm}

Theorem~\ref{thm:int} shows that the RSB distribution is a mixture of exponential distributions, which is conditionally conjugate as a prior for many sampling models, including the Poisson distribution. 
This result parallels the well-known mixture representation of the SB distribution (e.g., \citealt{armagan2011generalized}):
\begin{equation*}
    \pi_{\rm SB}(\eta) = \int _{\mathbb{R}} {\rm Ga}(\eta | a, u) {\rm Ga}(u|b,1)du,
\end{equation*}
which is also conditionally conjugate for the Poisson distribution.

In the augmented model in (\ref{RSB-integral}), note that the density of $v$ (and $w$) conditional on the other variables involves the reciprocal gamma function with argument $v+w+1$, from which it is not straightforward to simulate. 
However, the joint density of $(u,v,w)$ conditional on $\eta$ can be computed as 
\begin{equation*}
\begin{split}
        \pi (u,v,w|\eta) 
        & = \pi (u|v,w,\eta) \ \pi (v|\eta ) \ \pi (w|\eta ) \\
        & 
        = {\rm Ga}(u|v+w+1,1+\eta ) \ {\rm Ga}(v|1-a,\log(1+\eta )) \  {\rm Ga}(w|a+b,1+\log(1+\eta )).
\end{split}
\end{equation*}
The simulation from this composition form is straightforward. This means that the RSB distributions can replace the exponential/gamma distributions in the model of interest without complicating the posterior computation by using the Markov chain Monte Carlo method. In Section~\ref{sec:pos}, we use this result to provide an efficient posterior computation algorithm for the robust regression models for counts. 

For a better understanding of the proposed distribution, we consider the transformation defined by $\tilde{v}=v+w$ and $\tilde{w}=v/(v+w)$. Then, the RSB distribution is the marginal distribution of the following hierarchical model: 
\begin{equation*}
    \eta |u,\tilde{v},\tilde{w} \sim \mathrm{Ex}(u), \ \ \ u|\tilde{v},\tilde{w} \sim \mathrm{Ga}(\tilde{v},1), \ \ \ \tilde{v}|\tilde{w} \sim \mathrm{Ga}(b,\tilde{w}), \ \ \ \mathrm{and} \ \ \ \tilde{w}\sim \mathrm{Be}(a,1-a).
\end{equation*}
This expression resembles the SB distribution and its representation as a mixture of exponential distributions. That is, if $\eta \sim \mathrm{SB}(a,b)$, then it is written as the marginal distribution of the following hierarchical model (e.g., \citealt{zhang2022bayesian}): 
\begin{equation*}
    \eta |\tilde{v},\tilde{w} \sim \mathrm{Ex}(\tilde{v}), \ \ \ \tilde{v}|\tilde{w} \sim \mathrm{Ga}(b,\tilde{w}), \ \ \ \mathrm{and} \ \ \ \tilde{w}\sim \mathrm{Be}(a,1-a).
\end{equation*}
The difference from the RSB distribution is the lack of latent variable $u$. From this expression, the RSB distribution is regarded as an extension of the SB distribution by the additional hyperprior for the shape parameter of the gamma distribution. 

To summarize, the SB and RSB distributions can have the density spike at zero and (super) heavy tail, while being conditionally conjugate for the Poisson likelihood. 
The same density spike and heavy tail are also seen in the scale mixture of normals and utilized in the shrinkage estimation, where the scale follows the half-Cauchy distribution \citep{carvalho2009handling,carvalho2010horseshoe} or the variance follows the SB distribution \citep{armagan2011generalized,zhang2022bayesian}. Although these scale mixtures of normals can also be used for positive real values with truncation (e.g., \citealt{okano2022locally}), they are not conjugate for the Poisson likelihood.

%---------------------------------------------------------------------
%    Methods
%---------------------------------------------------------------------
\section{Robust Bayesian modeling of counts using the RSB distribution }\label{sec:method}

\subsection{Settings and models}
Let $y_i$ with $i=1,\ldots,n$ be count-valued response variables.
We consider the following generalized linear models: 
\begin{equation}\label{model}
y_i\sim {\rm Po}(\eta_i\lambda_i), \ \ \ \log\la_i=x_i^{\top}\beta + g_i^{\top}\xi_i + \log \gamma_i,
\end{equation} 
where $x_i$ and $g_i$ are the $p$- and $q$-dimensional vector of covariates, $\beta$ is a $p$-dimensional vector of regression coefficients, and $\xi_i$ is a $q$-dimensional vector of random effects. 
The second term, $g_i^{\top}\xi_i$, is typically used to model spatial and/or temporal effects. 
The last term, $\log \gamma _i$, extends the class of sampling models to the mixture of Poisson distributions. For example, if $\gamma _i$ follows a gamma distribution, then the marginal distribution of $y_i$ is the negative binomial distribution \citep[e.g.,][]{hilbe2011negative}. If $\gamma _i=1$, then the model reduces to the Poisson regression. This modeling philosophy is consistent with \cite{wang2018general}; expression (\ref{model}) can cover the variants of regression models for counts.

In model (\ref{model}), multiplicative noise $\eta_i$ is expected to absorb observations that cannot be explained by the structural term of $\la_i$. 
The standard Poisson model corresponds to a case in which $\eta_i=1$, but that model with outliers and zero inflation would result in inefficient and biased inference for $\beta$ and $\xi_i$. 
Another well-known model is zero-inflated Poisson regression, where $\eta_i\sim s\delta_0+(1-s)\delta_1$, with $\delta_a$ being the one-point distribution on $a$ \citep[e.g.,][]{lambert1992zero}.
The expected role of $\eta_i$ in robustness is that $\eta_i$ should have large values when $y_i$ is an outlier (large count) or $y_i/\la_i$ is very large. Likewise, $\eta_i$ should be sufficiently small when $y_i$ is a meaningless zero, or when the probability of $y_i=0$ is extremely small under the Poisson distribution with mean $\la_i$.
These properties of $\eta_i$ can automatically eliminate unnecessary information brought by zero inflation and large outliers, leading to robust posterior inference on $\la_i$. 
To achieve these properties, we use the SB and RSB distributions introduced in Section~\ref{sec:RSB}.

We propose the following distribution for $\eta_i$:
\begin{equation}\label{RSB}
f(\eta_i)=(1-s)\delta_1(\eta_i) + sH(\eta_i),  \ \ \ \ \eta_i>0,
\end{equation}
where $\delta_1(\cdot)$ is the one-point distribution on $1$, and $s\in (0,1)$ is an unknown mixing proportion. 
The first component is meant to capture non-outlying observations, whereas the second is designed to explain zero inflation and outliers. 
The distribution $H$ in the second component can be either SB or RSB distribution. Here, we consider the RSB distribution of the form, 
\begin{equation}\label{default-RSB}
H(\eta )=\frac1{\pi}\frac{\{\log(1+\eta )\}^{-1/2}}{(1+\eta )\{1+\log(1+\eta )\}}, \ \ \ \ \eta >0,
\end{equation}
which corresponds to density (\ref{RSB-general}) with $a=b=1/2$. 
As seen in Section~\ref{sec:RSB}, $H(\cdot)$ has an extremely heavy tail, $H(\eta )\sim \eta ^{-1}(\log \eta )^{-3 / 2}$ as $\eta \to \infty$, known as log-regularly varying density, and the value of the density diverges at the origin at the rate of $\eta ^{-1/2}$. 
Distribution (\ref{RSB}) has only one unknown parameter, $s$, that controls the potential number of outliers and zeros, and it is estimated by assigning a prior distribution.  
As long as $s>0$, the notable properties of $H$, the extremely heavy tail and spike at the origin, are inherited by mixture distribution (\ref{RSB}) for $\eta_i$. We refer to distribution (\ref{RSB}) as the {\it RSB mixture} distribution.

%   Posterior computation 
\subsection{Posterior computation}\label{sec:pos}
Another notable property of RSB mixture distribution (\ref{RSB}) is its computational tractability.
The posterior inference on $\la_i$ and $\eta_i$ under (\ref{model}) can be easily performed using a simple Gibbs sampling method. 
We first introduce the following hierarchical expression for $\eta_i$:
\begin{equation}\label{Eta}
    \eta_i = \begin{cases}
    \eta_{1i}\sim \delta_1 &\mathrm{if} \ z_i=0 \\
    \eta_{2i}\sim H(\cdot)  &\mathrm{if} \ z_i=1,
    \end{cases}
\end{equation}
where ${\rm Pr}(z_i=1) = 1 - {\rm Pr}(z_i=0) = s$, and $\eta_{1i}$ and $\eta_{2i}$ are mutually independent. 
Note that $z_i$ can be interpreted as an indicator of outliers; that is, $y_i$ can be an outlier if $z_i=1$.
Here, $s$ is an unknown mixing proportion, and we assign the conjugate prior $s\sim {\rm Beta}(a_s, b_s)$ with fixed hyperparameters $a_s$ and $b_s$. 
For the regression coefficients $\be$ and $\xi=(\xi_1,\ldots,\xi_n)$, let $\pi(\beta, \xi)$ be a prior distribution for $(\beta, \xi)$. Then, the joint posterior distribution is given by 
\begin{equation}\label{joint-pos}
\pi(\beta, \xi)s^{a_s-1}(1-s)^{b_s-1}\prod_{i=1}^n \left\{(1-s){\rm Po}(y_i;\lambda_i)\right\}^{1-z_i}
\left\{s{\rm Po}(y_i;\eta_{2i}\lambda_i)\right\}^{z_i}
H(\eta_{2i}),
\end{equation}
noting that $\lambda_i$ is a function of $(\beta, \xi)$.
The full conditional distributions of $\eta_{2i}$ are not well-known, but we can utilize the integral expression of the RSB distribution given in Theorem \ref{thm:int} to obtain the following hierarchical expression:
$$
\eta_{2i}|u_i\sim {\rm Ga}(1, u_i), \ \ 
u_i|v_i, w_i \sim {\rm Ga}(v_i+w_i, 1), \ \ 
(v_i,w_i)\sim g(v,w),
$$ 
where the joint distribution of $(v_i,w_i)$ is given by 
$$
g(v,w)=\frac1{\pi^{3/2}}\frac{v^{-1/2}e^{-w}}{v+w}.
$$
Because $\eta_{2i}$ is conditionally gamma-distributed, it follows from (\ref{joint-pos}) that the full conditional distribution of $\eta_{2i}$ is proportional to $\eta_{2i}^{y_i}\exp\{-\eta_{2i}(\lambda_i+u_i)\}$ if $z_i=1$. 
The Gibbs sampler can be derived for the augmented model with $(u_i,v_i,w_i)$, which is summarized as follows:

\medskip
%  Algorithm 
\begin{algo}[Sampling from the conditional posterior of $\eta_i$]
\label{algo:pos1}
Given $\lambda_i$, the conditional posterior sample of $\eta_i$ is obtained as $\eta_i=1-z_i + z_i\eta_{2i}$, where $z_i$ and $\eta_{2i}$ are generated using the following Gibbs sampler: 
\begin{itemize}
\item[-]
(Sampling of $z_i$) \ \  The full conditional distribution of $z_i$ is a Bernoulli distribution in which the probabilities of $z_i=0$ and $z_i=1$ are proportional to $(1-s){\rm Po}(y_i; \la_i)$ and $s{\rm Po}(y_i; \eta_{2i}\la_i)$, respectively. 

\item[-] 
(Sampling of $\eta_{2i}$) \ \  The full conditional distribution of $\eta_{2i}$ is ${\rm Ga}(y_i+1, \la_i+u_i)$ if $z_i=1$ and ${\rm Ga}(1, u_i)$ if $z_i=0$.

\item[-]
(Sampling from the other parameters and latent variables) \ \ The full conditional distributions of $u_i, v_i, w_i$, and $s$ are as follows: 
\begin{align*}
&s|\cdot \sim {\rm Beta}\left(a_s+\sum_{i=1}^nz_i, b_s+n-\sum_{i=1}^nz_i\right), \ \ \ w_i|\cdot \sim {\rm Ga}(1, 1+\log(1+\eta_{2i}))\\
&v_i|\cdot\sim {\rm Ga}(1/2, \log(1+\eta_{2i})), \ \ \ \ \  u_i|\cdot\sim{\rm Ga}(v_i+w_i+1, 1+\eta_{2i}).
\end{align*}
\end{itemize}
\end{algo}

Therefore, all the full conditional distributions related to local parameter $\eta_i$ are familiar forms. As such, the sampling steps can be efficiently carried out. 
Regarding the sampling of $(\beta,\xi)$, its full conditional distribution is proportional to 
\begin{equation}\label{Po-full}
\pi(\beta, \xi)\prod_{i=1}^n{\rm Po}(y_i;\exp(x_i^{\top}\beta+g_i^{\top}\xi_i+\log\gamma_i+\log\eta_i)).
\end{equation}
Hence, $(\beta,\xi)$ can be sampled via existing sampling algorithms by recognizing $\log\eta_i$ as an offset term. 
When the dimension of $(\beta,\xi)$ is not large, standard approaches such as the independent Metropolis-Hastings method can be adopted. 
As a more efficient strategy, we also suggest using an approximated likelihood in place of the Poisson likelihood. 
Let $\ep_i\sim {\rm Ga}(\delta, \delta)$ with a large $\delta>0$, and thus we have most of the probability mass in the neighborhood of $\ep_i=1$.
Then, we approximate the Poisson model (\ref{model}) as $y_i\sim {\rm Po}(\ep_i\la_i\eta_i)$, noting that the approximated model converges to original model (\ref{model}) as $\delta\to\infty$ because $\ep_i\to 1$ in probability. 
By marginalizing $\ep_i$ out and using Polya-gamma data augmentation \citep{polson2013bayesian}, the approximated model for $y_i$ can be expressed as
\begin{align*}
&\frac{\Gamma(y_i+\delta)}{\Gamma(\delta)y_i!}\frac{(\delta^{-1}\la_i\eta_i)^{y_i}}{(\delta^{-1}\la_i\eta_i+1)^{y_i+\delta}}\\
& \ \ \ \ \ \ 
=\frac{\Gamma(y_i+\delta)}{\Gamma(\delta)y_i!}2^{-(y_i+\delta)}e^{\kappa_i\psi_i}\int_0^{\infty}\exp\left(-\frac12\omega_i\psi_i^2\right)p_{PG}(\omega_i;y_i+\delta, 0)d\omega_i,
\end{align*}
where $\kappa_i=(y_i-\delta)/2$, $\psi_i=x_i^{\top}\beta + g_i^{\top}\xi_i +\log\gamma_i+\log\eta_i-\log\delta$, and $p_{PG}$ is the density function of the Polya-gamma distribution.  
Under the approximated model, the sampling steps for $(\beta,\xi)$ and $\omega_i$ are given as follows:

%  Algorithm 
\begin{algo}[Efficient sampling from the conditional posterior of $(\beta,\xi)$ using augmentation]
\label{algo:pos2}
Given $\eta_i$, the posterior samples of $(\beta,\xi)$ can be generated as follows: 
\begin{itemize}
\item[-]
(Sampling of $\beta$ and $\xi$) \  
The full conditional distribution of $(\beta,\xi)$ is proportional to
$$
\pi(\beta,\xi)\exp\left\{-\frac12\sum_{i=1}^n\omega_i(x_i^\top\beta+g_i^\top\xi_i)^2+\sum_{i=1}^n (\kappa_i+\omega_i\log(\delta/\gamma _i\eta_i))(x_i^\top\beta+g_i^\top\xi_i)\right\}.
$$
Hence, the multivariate normal distribution for $(\beta,\xi)$ is conditionally conjugate. 

\item[-]
(Sampling of $\omega_i$) \  
The full conditional distribution of $\omega_i$ is the Polya-gamma distribution: $PG(y_i+\delta, \psi_i)$. 
\end{itemize}
\end{algo}

Algorithm \ref{algo:pos2} is based on an approximation by the negative binomial likelihood and not exactly sampling from the posterior under the target model. If necessary, one can correct the potential bias by importance weighting the posterior samples based on the ratio of the target and approximate likelihoods. 
When the sampling from the Polya-gamma distribution is computationally costly, one can use the efficient sampling strategy of the Polya-gamma random variable based on normal approximation \citep{glynn2019bayesian}. Our model justifies the use of this approximation by setting $\delta$ to a large value. Finally, if one considers the negative binomial distribution by having $\gamma_i$ gamma distributed, then the algorithm above works as the exact posterior sampling without introducing $\epsilon_i$.

Our computational strategy developed here can also be applied to the SB mixture, where $H$ is the SB distribution, with minimal modification. For details, see Section~S7.

%   Theoretical Properties 
\subsection{Theoretical robustness properties}\label{sec:theory}
The advantage of the proposed method is its two-way protection property; the proposed method is robust against both zero inflation and large counts.
Here, we provide theoretical arguments regarding this property. 
In what follows, we assume $\gamma_i=1$ for simplicity.

We first address the zero-inflation issue.
This can be regarded as a situation where we observe $y_i=0$, even when the true regression surface $\lambda_i$ is sufficiently away from zero. Consequently, the probability of $y_i=0$ is extremely small under the true rate $\lambda _i$, which may favor smaller $\lambda_i$ and introduce bias in the posterior distribution of the parameters.  
In the proposed model with the SB and RSB mixture distributions, such an undesirable zero count can be absorbed by the additional term $\eta_i$ being very small, which limits the information from the observation in the posterior of the model parameters. 
Therefore, protection against zero inflation is related to the behavior of the posterior distribution of $\eta_i$ near the origin.
To study this distribution, we consider a general model by replacing the $H$ distribution in (\ref{RSB}) with either $\mathrm{SB}(a,b)$ or $\mathrm{RSB}(a,b)$ with shape parameters $a$ and $b$, whose density tails at the origin are given in (\ref{eq:zerodens}). 
The posterior distribution of $\eta_i$ is 
\begin{align}
p( \eta_i | y_i ) &\propto \eta_i^{y_i} e^{- \la _i \eta_i} \Big\{ (1 - s) \delta_1(\eta_i) + s H(\eta_i;a,b) \Big\} . 
\end{align}
Therefore, the posterior density behaves as $\eta_i^{y_i+a-1}$ near the origin (as $\eta_i\to 0$).
Hence, $\lim_{\eta_i \to 0} p( \eta_i | y_i ) = 0$ when $y_i \ge 1$; that is, the posterior mass is not concentrated near the origin when the observed count is non-zero. 
However, when $y_i = 0$, the shape of the posterior density function changes substantially depending on the value of $a$. The posterior density diverges if $a<1$ and converges to 0 when $a>1$. 
More precisely, we compare the posterior probability of $\eta_i$ near the origin under $y_i=0$, as presented in the following theorem.

% Theorem 
\begin{thm}\label{thm:zero}
For any $\ep>0$, the posterior probability $\Pr(\eta_i<\ep|y_i=0; a, b)$ is a non-decreasing function of $\lambda_i$. 
Moreover, if $0<a<1<a' < \infty$ and $s \in (0, 1]$, then $\Pr(\eta_i<\ep|y_i=0; a, b)/\Pr(\eta_i<\ep|y_i=0; a', b)\to\infty$ as $\ep\to 0$. 
\end{thm}

The first statement of the above theorem clarifies that the posterior probability concentrates more around the origin as the regression part moves farther away from $0$; that is, the probability of $y_i=0$ expected from the regression part becomes smaller. 
This suggests that in our model, $\eta_i$ can successfully absorb a meaningless zero count. 
In contrast, the second result shows that comparing the two values of the shape parameter, $a<1$ and $a'>1$, 
the posterior probability of $\eta _i \approx 0$ with $a$ is much larger than that with $a'$, implying that the choice of $a=1/2$ is reasonable for protection against zero inflation.

We next consider robustness against large outlying counts. 
For technical purposes, we replace the Dirac measure on $1$ in (\ref{Eta}) with a continuous approximation ${\rm Ga}(\alpha, \alpha)$ with a large fixed value $\alpha$. 
Note that the posterior distribution of $(\beta,\xi)$ is given by 
\begin{align}
p( \be ,\xi | y) &= \frac{ \pi ( \be,\xi ) \prod_{i = 1}^{n} \int_{0}^{\infty } g( \eta_i ) {\rm{Po}} ( y_i | \eta_i \exp( x_i^{\top}\be + g_i^{\top}\xi_i )) d{\eta_i} }{ \int \pi ( \be,\xi ) \big\{ \prod_{i = 1}^{n} \int_{0}^{\infty } g( \eta_i ) {\rm{Po}} ( y_i | \eta_i \exp( x_i^{\top}\be + g_i^{\top}\xi_i )) d{\eta_i} \big\} d\be d\xi } \text{.} \non 
\end{align}
We first define $\Kc$ and $\Lc$ as non-empty index sets of non-outliers and outliers, respectively. 
Similarly, $y_{\Kc } = ( y_i )_{i \in \Kc }$ and $y_{\Lc } = ( y_i )_{i \in \Lc }$ are the sets of non-outliers and outliers, respectively.
We demonstrate that under some suitable conditions, $p( \be , \xi | y) \to p( \be , \xi | y _{\Kc } )$ as outliers take extreme values. This asymptotic result cannot be achieved if the tail of $\eta_i$ is lighter than the RSB distribution (See Section~S4 of the Supplementary Materials), hence we assume $H$ to be the RSB distribution. 
A similar property has been proved under regression models for continuous responses \citep{Gag2019,hamura2020log}.

To be precise about the concept of (non-)outliers, we consider a framework in which $y_i$ is fixed for $i \in \Kc $, but $y_i = y_i ( \om ) \to \infty $ as $\om \to \infty $ for $i \in \Lc $.
Here, we assume that the covariate vectors are sufficiently linearly independent (see the Supplementary Material for more details on the technical assumption). 
Then, we have the following robustness property.

\begin{thm}
\label{thm:robust} 
Assume that $| \Kc _{+} | + 1 \ge | \Lc | + p$, where $\Kc _{+} = \{ i \in \Kc \ | \ y_i \ge 1 \}$ and $p$ is the dimension of $\beta$. 
Suppose that $\pi(\xi,\beta)=\pi(\xi)\pi(\beta)$, $\pi(\xi)$ is normal, and $\pi (\beta)$ is a proper probability density. 
Then, we have $p( \be , \xi | y) \to p( \be , \xi | y _{\Kc } )$ as $\om \to \infty$.
\end{thm}

Theorem~\ref{thm:robust} indicates that abnormally large counts, or outliers, are automatically ignored in the posterior distribution, even though we do not know which samples are outliers. 
To achieve this robustness property, it is essential that an error distribution is super heavily tailed. 
As previously stated, if the tail of the error distribution is lighter than that of the RSB mixture distribution, the posterior robustness does not hold. 
In contrast to the error distribution, there is no constraint on the choice of a prior distribution for $\beta$. It is noteworthy that the theorem proves the posterior robustness even when the density of $\beta$ is unbounded, and it allows the use of a shrinkage prior, such as the horseshoe prior (\citealt{carvalho2009handling,carvalho2010horseshoe}). The above robustness theorem excludes the use of improper priors for $\beta$. To guarantee the posterior robustness when using an improper prior for $\beta$, the condition of the number of non-zero counts must be strengthened; see Remark~S1 (in Section~S5) for details.

The RSB (mixture) distribution has no moment for its super-heavy tail, but the posterior of $(\beta ,\xi)$ has finite moments under mild conditions. 
\begin{prp}
\label{prp:posterior_moments} 
Let $m = \#\{ i \in \{ 1, \dots , n \} | y_i \ge 1 \}$ be the number of non-zero counts. 
\begin{enumerate}
    \item If $m\ge p$, the $k$-th posterior moment of $\xi$ is finite for any $k\ge 1$. 
    
    \item If $m\ge p + k - 1$ for some $k\ge 1$, the $k$-th posterior moment of $\beta$ is finite. 
\end{enumerate}
\end{prp}
The posterior means or variances of $\beta$ and $\xi$, which are frequently used for inference in practice, are guaranteed to exist. As proved in the Supplementary Material, this proposition is obtained as a corollary of the more general statement. The condition of the number of non-zero observations is similar to those necessary for finite posterior moments in robust linear regression models \citep[e.g.,][]{Gag2019, hamura2020log}.

%---------------------------------------------------------------%
%     Illustrations
%---------------------------------------------------------------%
\section{Illustrations}\label{sec:exm}
We demonstrate robust count modeling with an RSB mixture under a variety of frequently used models through both simulation and empirical studies. The details of the posterior computation algorithm are provided in Supplementary Material~S7.

%   Poisson regression 
\subsection{Robust count regression}
We first consider and investigate the performance of a standard count regression model with the RSB mixture distribution, together with existing robust and non-robust models. 
To this end, we generated $n=300$ observations from a Poisson regression model with $p=15$ covariates, given by $y_i^{\ast} \sim {\rm Po}(\la_i)$ for $i=1,\ldots,n$, where
$$
\log \la_i = \beta_0 +\sum_{k=1}^p \beta_kx_{ik}.
$$
We set $\beta_0=0.5$, randomly generate $\beta_1, \beta_2, \beta_3$ from $U(0, 0.4)$, $U(0.3, 0.7)$ and $U(0.1, 0.5)$, respectively, in each replication, and set the other coefficients to $0$.
Here, the vector of covariates $(x_{i1},\ldots,x_{ip})$ was generated from a multivariate normal distribution with a zero-mean vector and variance--covariance matrix with $(k,\ell)$-element equal to $(0.2)^{|k-\ell|}$ for $k,\ell \in \{1,\ldots,p\}$.
For each $i$, we independently generated a discrete random variable $d_i$, following $\Pr[d_i=1]=\omega_1$, $\Pr[d_i=2]=\omega_2$, and $\Pr[d_i=0]=1-\omega_1-\omega_2$, and we then generated the contaminated observations as 
\begin{equation*}
    y_i = \begin{cases}
    y_i^{\ast} & \mathrm{if} \ d_i=0 \\
    0 & \mathrm{if} \ d_i = 1 \\
    y_i^{\ast}+y_o & \mathrm{if} \ d_i=2 
    \end{cases}.
\end{equation*}
Thus, $d_i=1$ and $d_i=2$ correspond to a meaningless zero and a large outlying count, respectively.
The severity of outliers is controlled by $y_o$; we considered two values, $y_o=20$ and $y_o=50$. 
We also adopted the following nine cases for the combinations of $\omega_1$ and $\omega_2$.
$$
\begin{array}{cccccccccccc}
{\rm Scenario} & 1 & 2 & 3 & 4 & 5 & 6 & 7 & 8  \\
\omega_1& 0 & 0.05 & 0.05 & 0.05 & 0.05 & 0.1 & 0.1 & 0.1   \\
\omega_2& 0 & 0 &  0.05 & 0.1 & 0.15 & 0.05 & 0.1 & 0.15
\end{array}
$$

For the simulation dataset, we applied Poisson regression using RSB mixture distribution (\ref{RSB}) with the default choice of hyperparameters, denoted by RSB.
We also applied the model with SB distribution, where the hyperparameters are set to $a=1/2$ and $b=1/10$ to make the tail of the SB distribution heavy to approximately achieve posterior robustness. 
For comparison with existing models, we applied the standard Poisson regression (PR), negative binomial regression (NB), zero-inflated Poisson regression (ZIP), and zero-inflated negative binomial regression (ZINB) models, where all the methods are implemented in a Bayesian way. 
As a prior distribution, we assign $\beta=(\beta_0,\beta_1,\ldots,\beta_p)\sim N(0, 100I_{p+1})$, and employed an independent Metropolis-Hasting (MH) algorithm to sample $\beta$ from its full conditional distribution. 
We generated 1000 posterior samples, after discarding the first 1000 posterior samples, to compute posterior means and $95\%$ credible intervals of $\beta_k$ for $k=1,\dots, p$.

To evaluate the estimation performance, we computed the mean squared errors (MSE) of the point estimates of $\beta$, scaled mean squared errors (SMSE) of the Poisson intensity, and interval scores (IS) given by \cite{gneiting2007strictly} of the $95\%$ credible intervals of $\beta$, based on $R=500$ replications of the simulation.
These three measures are defined as 
$$
{\rm MSE}=\frac{1}{(p+1)R}\sum_{r=1}^R\sum_{k=0}^p (\beh_{k(r)}-\beta_k)^2, \ \ \ \ 
{\rm SMSE}=\frac{1}{nR}\sum_{r=1}^R\sum_{i=1}^n \frac{(\lah_{i}-\la_i)^2}{\la_i^2}
$$
and
$$
{\rm IS}=\frac1R\sum_{r=1}^R\sum_{k=0}^p\left[{\rm CI}_{k(r)}^u-{\rm CI}_{k(r)}^l+\frac2{\alpha}\Big\{(\beta_k-{\rm CI}_{k(r)}^u)_+ + ({\rm CI}_{k(r)}^l-\beta_k)_+\Big\}\right], 
$$
where $\beh_{k(r)}$ is a point estimate of $\beta_k$, $\lah_i=\exp(x_i^\top \beh)$ is estimated intensity, and ${\rm CI}_{k(r)}^u$ and ${\rm CI}_{k(r)}^l$ are the upper and lower values of the credible interval of $\beta_k$ in the $r$th replication. 
Here, $\alpha$ is the nominal level (e.g., $\alpha=0.05$) and $(x)_{+}=\max(0, x)$.
These values were averaged over $k=1,\ldots, p$.
Note that IS is a scalar measure that takes into account both the coverage probability and interval length; the smaller the IS, the higher the coverage rate and narrower the interval estimates.
The results are shown in Figure~\ref{fig:sim-nst}, in which the approximated $99\%$ confidence intervals of MSE and IS, based on the estimated Monte Carlo errors, are also presented.

In scenarios 1 and 2, the generated data are zero inflated but have no large outliers. 
Hence, the ZIP and ZINB models show reasonable performance in terms of both point and interval estimation, as expected, and the proposed RSB and SB are comparable to these ideal methods. 
In scenarios 3-8, however, large outliers are included in the simulated data, and the MSEs of the ZIP and ZINB models rapidly increase, whereas those of the proposed RSB and SB methods remain almost the same, which clearly indicates the strong robustness of the proposed method.
Comparing RSB and SB, it is observed that the performance of RSB tends to be better than that of SB under existence of a considerable number of outliers.
This would be because SB does not hold posterior robustness unlike RSB as discussed in Section~\ref{sec:pos}.
It is also observed that the amount of improvement of the proposed method in terms of MSE is more substantial under $y_o=50$ than under $y_o=20$.
In terms of interval estimation, we can observe similar tendencies in MSE and IS; RSB provides stable performance regardless of the existence of zero inflation or outliers, while the interval estimation using the other methods is significantly worse in the presence of outliers. 
This observation is consistent with Theorem~\ref{thm:robust}, which guarantees the robustness of the entire posterior distribution, including the credible intervals.  
We can also observe from the results of the IS that the over-dispersed models (NB and ZINB) provide better interval estimation than PR and ZIP in the presence of outliers, implying the importance of heavily tailed error distribution in uncertainty quantification. However, as the RSB model performs better than these methods in terms of IS, introducing over-dispersion is not necessarily successful in dealing with local outliers.
This limitation of the robustness of NB is consistent with our theoretical result given in Theorem \ref{thm:robust}; the negative binomial distribution is not heavily tailed enough to account for large outliers. 
Although we adopted the default choice of the hyperparameter $(a,b)=(1/2, 1/2)$ in RSB, the results would not be sensitive to the choice of these hyperparameters since posterior robustness holds regardless of hyperparameters.
To show such insensitivity, we provide results of RSB under various choices of hyperparameters in Supplementary Materials.

%  Figure
\begin{figure}[!htb]
\centering
\includegraphics[width=12.5cm,clip]{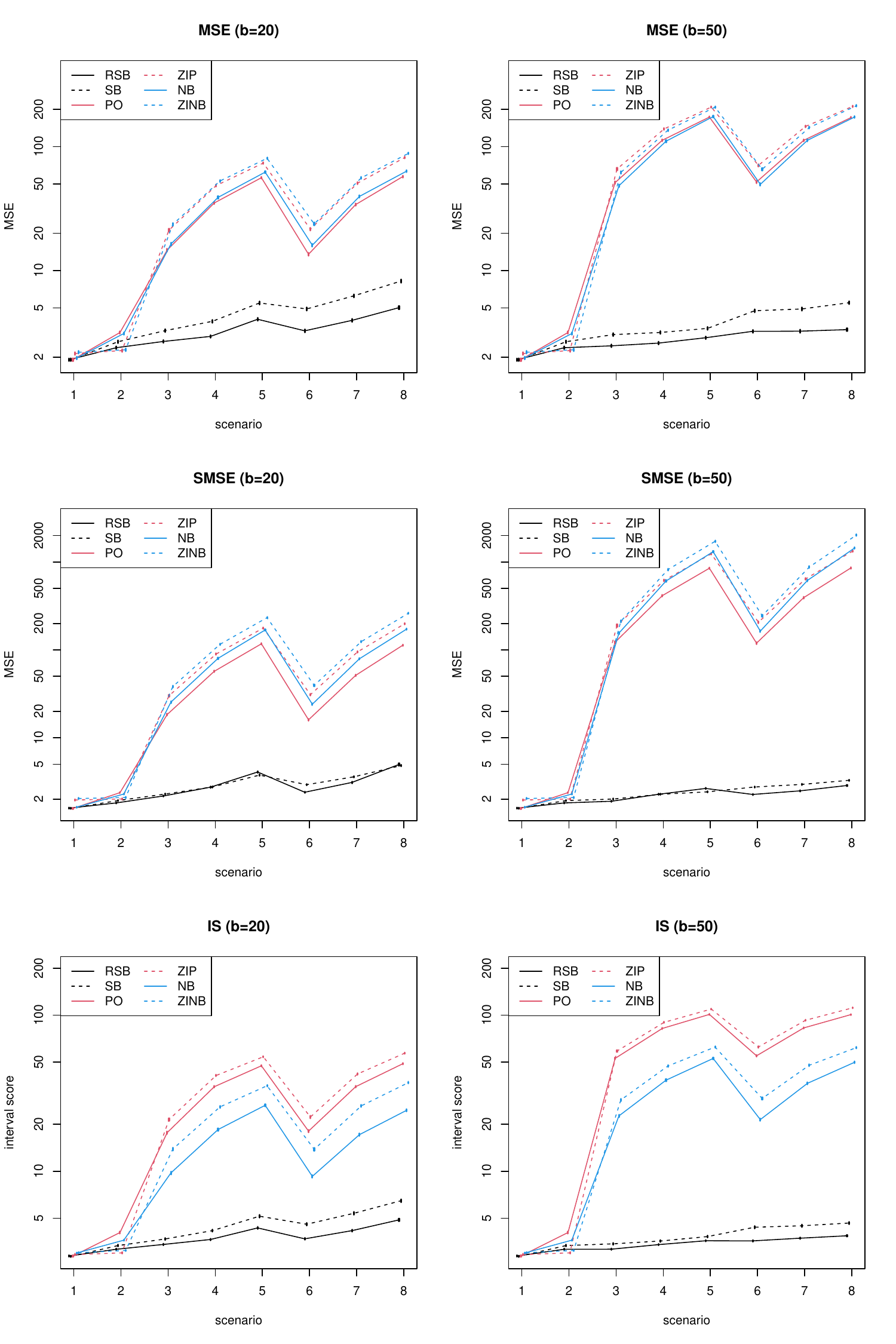}
\caption{Simulated mean squared errors (MSE), scaled mean squared errors (SMSE) and average values of interval score (IS) of the Poisson regression (PO), negative binomial regression (NB), zero-inflated Poisson (ZIP) regression, zero-inflated negative binomial (ZINB) regression, and RSB and SB regression models under 8 scenarios. 
The vertical lines are approximated $99\%$ confidence intervals of the MSE, SMSE and IS based on the estimated Monte Carlo errors. 
\label{fig:sim-nst}
}
\end{figure}

%   trend filtering
\subsection{Robust locally adaptive smoothing}
We next consider locally adaptive smoothing in the presence of outliers.
Suppose that $n$ count observations $y_i, \ i=1,\ldots,n$, are independent, and its expectation can be expressed as $f(x_i)$ as a function of some covariate $x_i$.
To estimate $f(\cdot)$, we assume that $y_i\sim {\rm Po}(\eta_i\lambda _i)$ and $\log \lambda _i =  f(x_i) $, where $\eta_i$ follows RSB mixture distribution (\ref{RSB}). 
For simplicity, we further assume that $x_i=i$ (i.e., the covariate is equally spaced).
Then, by letting $\theta_i=f(x_i)$, we introduce the following hierarchical model for adaptive smoothing: 
\begin{equation}\label{Po-TF}
\begin{split}
&y_i|\theta_i\sim {\rm Po}(\eta_i e^{\theta_i}), \ \ \ i=1,\ldots,n,\\ &D^k\theta_j\sim {\rm HS}(\tau^2),\ \ \ j=1,\ldots,n-k,
\end{split}
\end{equation}
where $D^k$ is the order-$k$ forward differences, (e.g. , $D^1\theta_j=\theta_{j+1}-\theta_j$) and HS$(\tau^2)$ is the horseshoe (HS) prior \citep{carvalho2010horseshoe} with an unknown global scale parameter $\tau^2$.
Note that \cite{faulkner2018locally} adopted an approximate version of the HS prior for $D^k\theta_j$; however, in (\ref{Po-TF}), we apply the original HS prior.
We note that the independent prior for $D^k\theta_j$ leads to a rank-deficient multivariate prior for $(\theta_1,\ldots,\theta_n)$.
Here, $\eta_i$ in (\ref{Po-TF}) is an independent random variable following the RSB mixture distribution to absorb zero inflation and outliers.
We adopt the half-Cauchy prior for $\tau$.
Then, the HS prior in (\ref{Po-TF}) can be expressed in the following hierarchy (e.g., \citealt{makalic2015simple}): for $j=1,\ldots,n-k$, 
\begin{equation*}
    D^k\theta_j\sim N(0, \phi_j\tau^2), \ \ \ \ \phi_j|\psi_{j}\sim {\rm IG}(1/2, 1/\psi_{j}), \ \ \ \ \psi _j \sim {\rm IG}(1/2, 1), 
\end{equation*}
and 
\begin{equation*}
\tau^2|\gamma\sim {\rm IG}(1/2, 1/\gamma), \ \ \ \ \gamma\sim {\rm IG}(1/2, 1).   
\end{equation*}

We conduct simulation studies to investigate the performance of the proposed model and a non-robust counterpart with $\eta_i=1$ in (\ref{Po-TF}), corresponding to the standard locally adaptive smoothing for count data.  
The following four scenarios for the true trend functions are considered. 
\begin{align*}
&{\rm  Scenario\ 1}:  \exp(\theta_i)=5,  \\
&{\rm  Scenario\ 2}:  \exp(\theta_i)=5-2I(i\geq 20)+5I(i\geq 40)-4I(i\geq 60), \\
&{\rm  Scenario\ 3}:  (e^{\theta_1},\ldots,e^{\theta_n})\sim N(5, \Sigma), \ \ \Sigma_{j,k}= 2\exp\left\{-\frac{(j-k)^2}{200}\right\},\\
&{\rm Scenario\ 4}: \exp(\theta_i)=5+\frac{5}{2}\sin\Big(\frac{4i}{n}-2\Big)+5\exp\Big\{-30\Big(\frac{4i}{n}-2\Big)^2\Big\}.
\end{align*}
The four scenarios correspond to constant, piecewise constant, smooth trend, and varying smoothness. This setting is almost the same as that adopted in \cite{faulkner2018locally}.
The genuine observations are generated from ${\rm Po}(e^{\theta_i})$. Then, for each observation, we add 40 with a probability of 0.1, or set it to 0 with a probability of 0.05, to create large outlying counts and zero inflation. 
We apply the robust method with the RSB mixture, SB mixture and standard method using the negative binomial (NB) models; we used $k=2$ (i.e., second-order differences) in all scenarios.

Using 1000 posterior samples after discarding the first 1000 samples of the proposed robust method with the RSB mixture as well as the SB and NB models, we computed the posterior mean and point-wise $95\%$ credible intervals of $e^{\theta_i}$.
The results are shown in Figure \ref{fig:sim-TF}.
It is observed that the NB model is sensitive to outliers, which prevents stable estimation of the true signals. 
In contrast, by introducing the RSB mixture to absorb potential outliers and zero inflation, the proposed method can produce reasonable point estimates and credible intervals that effectively cover the true function. 
The SB mixture provides similar, but slightly different posterior estimates to those of the RSB mixture.

%  Figure
\begin{figure}[!htb]
\centering
\includegraphics[width=14cm,clip]{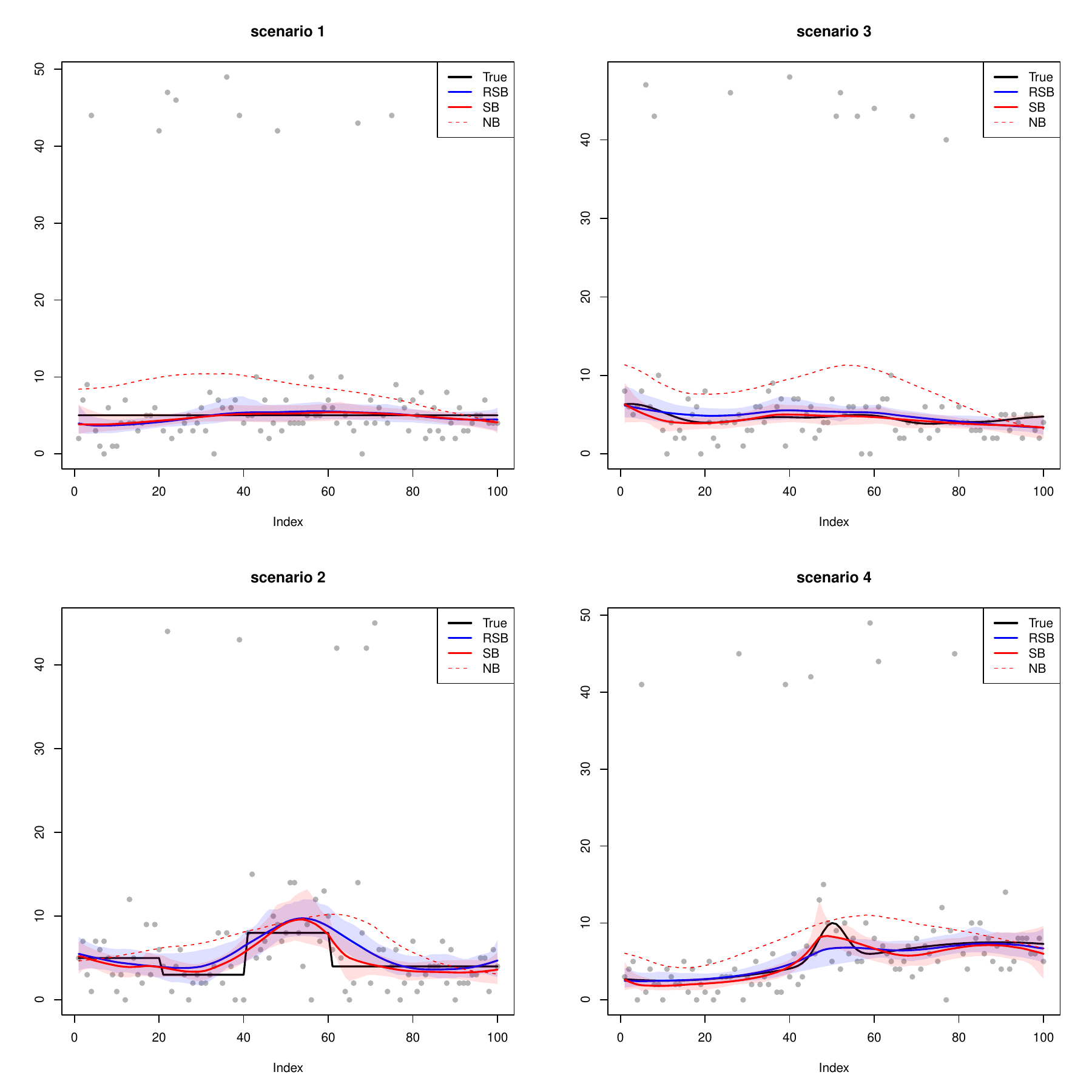}
\caption{The posterior medians of the underlying signals based on Bayesian trend filtering with negative binomial (NB), SB and RSB distributions. The shaded regions are point-wise $95\%$ credible intervals of the SB and RSB models.
\label{fig:sim-TF}
}
\end{figure}

%   spatial model
\subsection{Robust hierarchical count regression with spatial effects}
Next, we apply spatial count regression models with the RSB mixture distribution to analyze crime data. 
The response variable is the number of violent crimes in $n=2791$ local towns (equipped with the longitude and latitude of the main location) in the Tokyo metropolitan area in 2017, available at ``GIS database of number of police-recorded crime at O-aza, chome in Tokyo, 2009-2017'' of University of Tsukuba (\url{https://commons.sk.tsukuba.ac.jp/data_en}).
For auxiliary information about each town, we adopted area (km$^2$), population densities at noon and night, density of foreign people, percentage of single-person households, and average duration of residence, which can be constructed from census 2015, available at Japanese Government Statistics (\url{https://www.e-stat.go.jp/en}) and Statistics of Tokyo (\url{https://www.toukei.metro.tokyo.lg.jp}).
The processed dataset is provided in the Supplementary Material. 
For $i=1,\ldots,n$, let $y_i$ be the observed count of violent crimes, $a_i$ be the area, $x_i$ be the five-dimensional vector of standardized auxiliary information, and $s_i$ be the two-dimensional vector of location information in the $i$th local town.
In Figure \ref{fig:sp-sample}, we provide scatter plots of the observed number of crimes in each location; the right panel restricts the range to $(0,20)$.
The right panel shows that there are many zero counts, and we find that 37\% of $y_i$ is $0$. 
The left figure shows that there are some extremely large counts as potential outliers.

We employ the following spatial count model to estimate the structure of the number of crimes: 
\begin{equation}\label{crime-model}
y_i\sim {\rm Po}(\lambda_i \eta_i),  \ \ \log \lambda_i=x_i^{\top}\beta +\log a_i + \xi(s_i) + \log\gamma_i, \ \ \ i=1,\ldots,n,
\end{equation}
where $\eta_i$ follows RSB mixture (\ref{RSB}), $\gamma_i=1$, and $\xi_i$ is an unobserved spatial effect.
Note that $\log a_i$ is interpreted as an offset term, and $x_i^{\top}\beta$ and $\xi(s_i)$ can be seen as fixed and random effects on the logarithm of the crime rate per km$^2$.  
In the following, we refer to model (\ref{crime-model}) as the RSB model. 
We model spatial effect $\xi(s_i)$ using the predictive Gaussian process \citep{banerjee2008gaussian}; we set $\xi(s_i)=D_h(s_i)^\top\mu_{\xi}$, where $\mu_{\xi}$ is the standard Gaussian process on a set of fixed knots, $\kappa_1,\ldots,\kappa_M$, and $D_h(s_i)$ is an $M$-dimensional vector. 
Specifically, we assume that $\mu_{\xi}\sim N(0, \tau^{-1}H_h)$ with unknown precision parameter $\tau$ and variance--covariance matrix $(H_h)_{k\ell}=\nu_h(\kappa_k, \kappa_\ell)$.
Here, $D_h(s_i)$ is written as a regression, and $D_h(s_i)=H_h^{-1}V_h(s_i)$, where $V_h(s)=(\nu_h(s, \kappa_1), \ldots,\nu_h(s, \kappa_M))^\top\in \mathbb{R}^{M}$ and $\nu_h(s_1, s_2)$ is a correlation function with bandwidth $h$.
In this example, we use the Gaussian correlation function $\nu_h(s_1, s_2)=\exp\{-\|s_1-s_2\|^2/h^2\}$.
Then, the joint distribution of $\xi=(\xi(s_1),\ldots,\xi(s_n))^\top$ is $N(0,\tau^{-1}D_h^\top H_hD_h)$, where $D_h=(D_h(s_1),\ldots,D_h(s_n))$ is an $M\times n$ matrix.  
We set $M=100$, and the knot locations are set to the center points of $M$ clusters obtained by applying the k-means to the sampled locations $\{s_1,\ldots,s_n\}$. 
For comparison, we also consider the SB mixture ($\eta_i$ follows the SB mixture with $(a,b)=(1/2, 1/10)$ and $\gamma_i=1$) and the negative binomial (NB) regression ($\eta _i = 1$, $\gamma _i\sim {\rm Ga}(\nu,1)$ and $\nu \sim \mathrm{Ga}(1,1)$).

We generated 10000 posterior samples after discarding the first 5000 samples for the three models.
In Table~\ref{tab:sp-criterion}, we report the DIC \citep{spiegelhalter2002bayesian}, WAIC \citep{watanabe2010asymptotic} and PPL \citep{gelfand1998model} of the three models. 
The results show that the three comparison measures of RSB are considerably smaller than those of the NB model, which shows that the proposed models fit substantially better to the dataset.
Since the measures of RSB are smaller than those of SB, we report the results of RSB only in what follows. 
The posterior medians of $\eta_i$ and the trace plot of $s$ under RSB are shown in Figure~\ref{fig:sp-eta}.
The posterior distribution of $s$ is located away from $0$, and we observed several posterior medians of $\eta_i$ with values of more than $10$ or less than $0.1$. As evident from these results, under the RSB model, some observations are explained by error term $\eta_i$, not by auxiliary information or local spatial effects, implying possible abnormal zeros and outliers in this dataset. 
The posterior means of spatial effects $\xi_i$ in (\ref{crime-model}) are presented in Figure~\ref{fig:sp-effect}, along with locations whose posterior medians of $\eta_i$ are larger than $10$ or smaller than $0.1$. 
It can be observed that, under the NB model, the spatial effect is greatly influenced by outliers. In other words, the NB model explains a small number of large counts in the dataset by the spatial effect terms as a regional phenomenon. 
The problem of fitting the spatial effect to outliers is the spillover of extremely high estimates, for which the estimates in the neighboring areas of outliers are biased.
Therefore, it is more appropriate to assume that the large counts should not be explained by smoothly varying spatial effects, but by local random variables such as $\eta_i$ in the RSB model. 
In fact, the proposed RSB model can successfully eliminate information from extremely large counts, as evidenced by the spatially smoothed estimates of the spatial effect. 
The posterior means and $95\%$ credible intervals of the regression coefficients $(\beta _0,\dots,\beta_5)$, precision parameter of the Gaussian process $\tau$, and bandwidth of the correlation function $h$ are reported in Table~\ref{tab:sp-beta}. 
We can confirm that the RSB and NB models report different posteriors for some regression coefficients. 
For example, the estimates of $\beta_1$ (coefficients of night population density) are underestimated in the NB model compared with the RSB model, possibly because of the outliers. 
Further, the same phenomenon happens for $\beta_5$.
We also note that the posterior means of the spatial bandwidth parameter $h$ in the NB model is much smaller than that of the RSB model.
This is because under the NB model, the spatial effect term $\xi_i$ is affected by outliers and zero inflation, and the spatial effect is forced to change abruptly over the space, as shown in Figure~\ref{fig:sp-effect}, with a small bandwidth value.

%  Table 
\begin{table}[htbp!]
\caption{
The deviance information criterion (DIC), widely applicable information criterion (WAIC) and posterior predictive loss (PPL) of the three models. 
}
\label{tab:sp-criterion}
\begin{center}
\begin{tabular}{ccccccccccc}
\hline
 && RSB && SB && NB \\
\hline
DIC && 8518 && 8991 && 9568  \\
WAIC &&  1.66 && 1.74 && 1.75 \\
PPL && 4.35 &&  4.75 && 43.11  \\
\hline
\end{tabular}
\end{center}
\end{table}

%  Figure
\begin{figure}[!htb]
\centering
\includegraphics[width=14cm,clip]{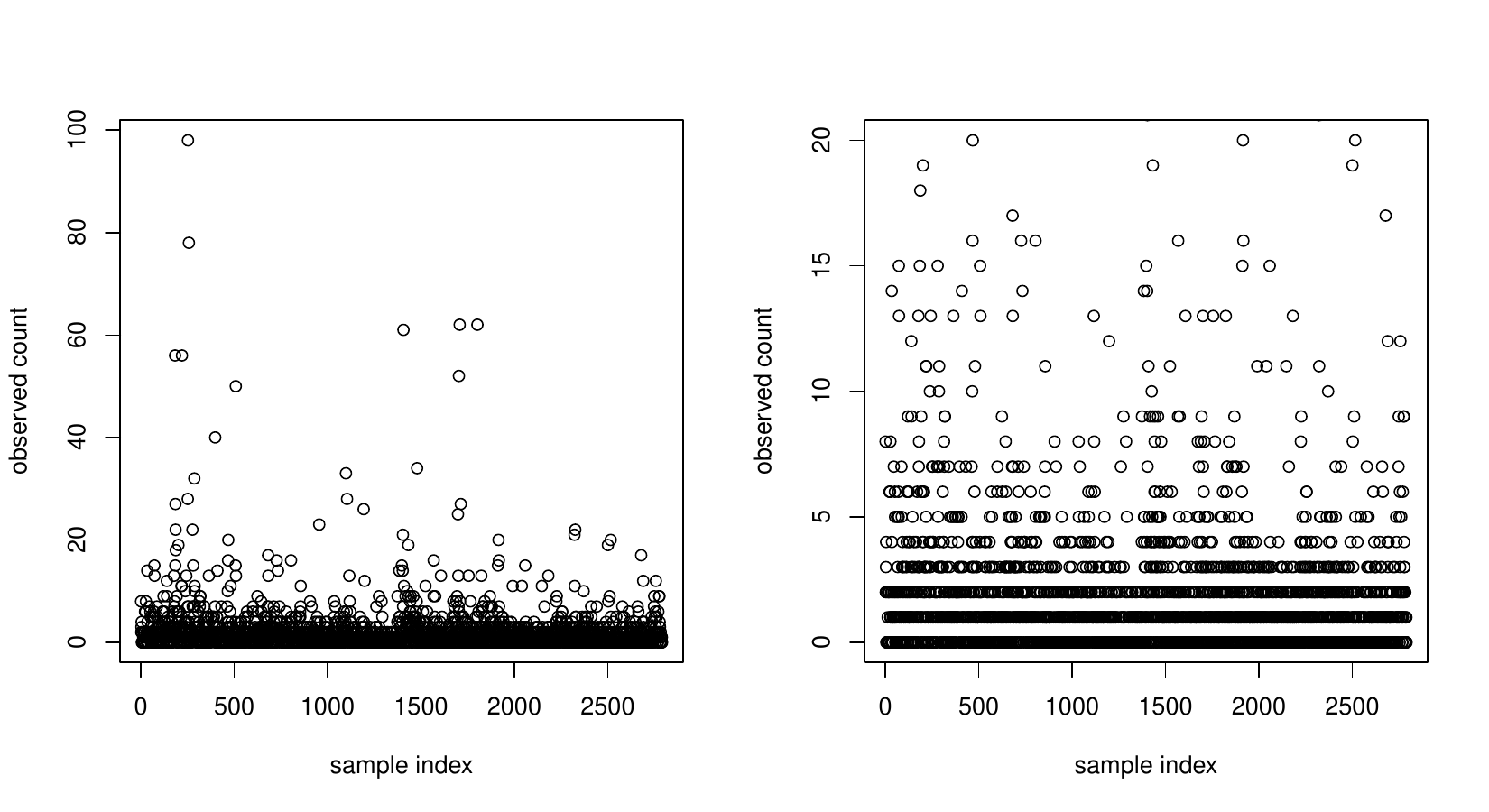}
\caption{Scatter plot of the observed number of crimes, where the range of the right figure is restricted to $(0, 20)$.
\label{fig:sp-sample}
}
\end{figure}

%  Figure
\begin{figure}[!htb]
\centering
\includegraphics[width=14cm,clip]{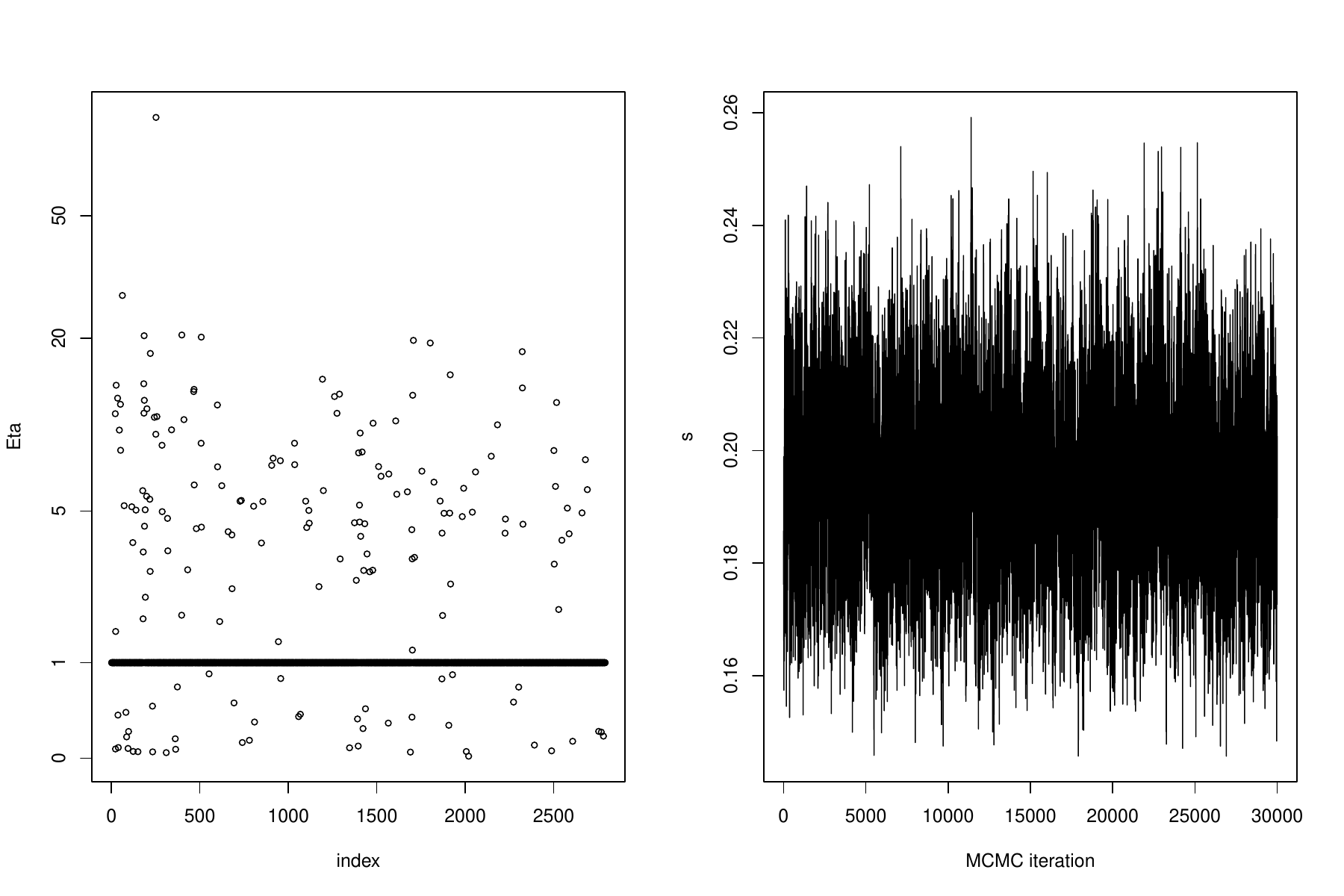}
\caption{The posterior median of $\eta_i$ (left) and the trace plot of the mixing proportion $s$ in the RSB mixture distribution (right). 
\label{fig:sp-eta}
}
\end{figure}

%  Figure
\begin{figure}[!htb]
\centering
\includegraphics[width=15cm,clip]{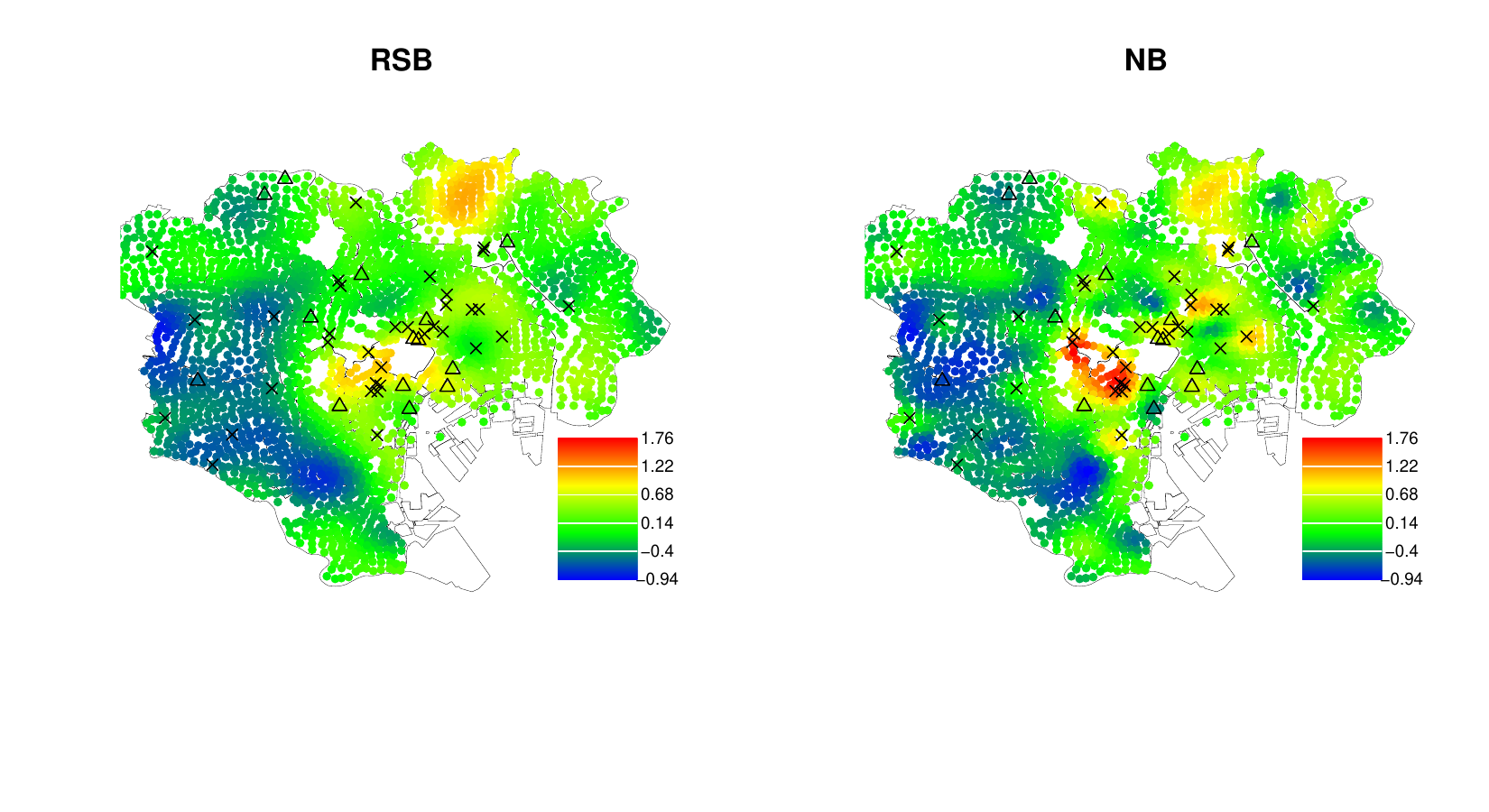}
\caption{The posterior means of the spatial effects under the RSB and NB models. 
The sampling points where the posterior median of $\eta_i$ under the RSB model are greater than 10 or smaller than $0.1$ are indicated by $\times$ and $\triangle$, respectively. 
}
\label{fig:sp-effect}
\end{figure}

%  Table 
\begin{table}[htbp!]
\caption{Posterior means and $95\%$ credible intervals of regression coefficients, scale, and bandwidth parameters in the spatial effects.
}
\label{tab:sp-beta}
\begin{center}
\begin{tabular}{lcccccccccc}
\hline
&&\multicolumn{3}{c}{RSB} & & \multicolumn{3}{c}{NB}\\
&&posterior & \multicolumn{2}{c}{credible interval} && posterior & \multicolumn{2}{c}{ credible interval}\\
&  & mean & Lower & Upper &  & mean & Lower & Upper \\
\hline
$\beta_0$ (Intercept) &  & 1.74 & 1.87 & 1.94 &  & 1.73 & 1.90 & 2.06 \\
$\beta_1$ (NPD) &  & 0.20 & 0.26 & 0.33 &  & 0.00 & 0.08 & 0.15 \\
$\beta_2$ (DPD) &  & -0.14 & -0.06 & 0.00 &  & -0.09 & 0.00 & 0.08 \\
$\beta_3$ (FD) &  & -0.01 & 0.05 & 0.11 &  & -0.01 & 0.06 & 0.14 \\
$\beta_4$ (SH) &  & 0.41 & 0.48 & 0.54 &  & 0.52 & 0.59 & 0.65 \\
$\beta_5$ (AYL) &  & -0.11 & -0.06 & -0.01 &  & -0.10 & -0.04 & 0.02 \\
$h$ &  & 0.31 & 0.41 & 0.50 &  & 0.17 & 0.22 & 0.30 \\
$\tau$ &  & 2.98 & 4.56 & 6.62 &  & 2.02 & 2.96 & 4.18 \\
\hline
\end{tabular}
\end{center}
{\footnotesize *NPD: night population density; DPD: daytime population density; FD: density of foreign people; SH: percentage of single-person households; AYL: average year of living }
\end{table}

\vspace{-1cm}

%-----------------------------------------------------------------
%       Discussion
%-----------------------------------------------------------------
\section{Concluding remarks}\label{sec:remark}
This paper introduced a new approach to robust hierarchical modeling of counts in the presence of both zero inflation and large outliers.
A key tool of the proposed approach is a new family of distributions called RSB distributions, which are used to absorb the effects of zero inflation and outliers, making the estimation of model structures in the Poisson rate robust.
There are two notable features of the proposed approach. 
First, the posterior computation is simply carried out via a costumed Gibbs sampler. 
Second, we can demonstrate the theoretical posterior robustness, which guarantees valid posterior estimation as well as inference under the existence of outliers.

\section*{Acknowledgement}
This work was supported by the Japan Society for the Promotion of Science (JSPS KAKENHI) grant numbers 20J10427, 18H03628, and 21H00699.

\vspace{1cm}
%   Reference
\bibliographystyle{chicago}
\bibliography{ref}

\newpage
%-----------------------------------------------------------------%
%         Supplement Material                                     %
%-----------------------------------------------------------------%
\setcounter{page}{1}
\setcounter{equation}{0}
\renewcommand{\theequation}{S\arabic{equation}}
\setcounter{section}{0}
\renewcommand{\thesection}{S\arabic{section}}
\setcounter{lem}{0}
\renewcommand{\thelem}{S\arabic{lem}}
\setcounter{thm}{0}
\renewcommand{\thethm}{S\arabic{thm}}
\setcounter{table}{0}
\renewcommand{\thetable}{S\arabic{table}}
\setcounter{figure}{0}
\renewcommand{\thetable}{S\arabic{figure}}
\setcounter{prp}{0}
\renewcommand{\theprp}{S\arabic{prp}}
\setcounter{remark}{0}
\renewcommand{\theremark}{S\arabic{remark}}

\def\Ht{{\tilde H}}
\def\al{{\alpha}}
\def\be{{\beta}}
\def\ga{{\gamma}}
\def\de{{\delta}}
\def\ep{{\varepsilon}}
\def\la{{\lambda}}
\def\si{{\sigma}}
\def\om{{\omega}}
\def\th{{\theta}}
\def\fai{{\varphi}}
\def\ka{{\kappa}}
\def\ze{{\zeta}}
\def\Th{{\Theta}}
\def\De{{\Delta}}
\def\Si{{\Sigma}}
\def\Ga{{\Gamma}}
\def\Om{{\Omega}}
\def\Up{{\Upsilon}}
\def\La{{\Lambda}}

\begin{center}
{\LARGE {\bf Supplementary Material for ``Robust Bayesian Modeling of Counts under Zero-inflation and Outliers: Theoretical Robustness and Efficient Computation"}}
\end{center}

\vspace{1cm}

This Supplementary Material provides technical proofs of Theorems 1-3 and Proposition 1, conditions for posterior robustness, detailed posterior computation algorithms for models used in Section~4 and additional numerical results. 

\section{Proofs and derivations for Section~1}

\subsection*{Proof of Theorem~1}

We first review the RSB density provided in Equation~(1) of the main text: 
\begin{equation*}
\pi_{\rm RSB}(\eta;a,b)=\frac{1}{B(a,b)}\frac{\{\log(1+\eta)\}^{a-1}}{1+\eta}\frac{1}{\{1+\log(1+\eta)\}^{a+b}}, \ \ \ \eta>0.
\end{equation*}
This density is defined for any positive $a$ and $b$, but in Theorem~1 we assume that $0<a<1$. 
The equation we want to show is Equation~(3) of the main text: 
\begin{equation*}
\begin{split}
\pi _{\rm RSB}(\eta;a,b) &= \int_{\mathbb{Re}_+^3} {\rm Ex}(\eta|u) \ {\rm Ga}(u|v+w,1) \ g(v,w) \ dudvdw \\
&= \int_{\mathbb{Re}_+^3} ue^{-u\eta} \ \frac{u^{v+w-1}e^{-u} }{\Gamma (v+w)} \ \frac{1}{B(a,b)\Gamma (1-a)\Gamma (a+b)} \frac{v^{-a}w^{a+b-1}e^{-w} }{ v+w } dudvdw.
\end{split}
\end{equation*}
To prove this, we utilize the integral representation of the normalizing constant of the gamma distribution. That is, for any positive $\alpha$ and $\beta$, we have
\begin{equation*}
    \beta ^{-\alpha} = \int _0^{\infty} \frac{t^{\alpha -1} e^{-\beta t} }{\Gamma(\alpha)} dt.
\end{equation*}
We apply this formula three times to the RSB density. First, we augment the term of $\{ 1+\log(1+u) \}$ as 
\begin{equation*}
\begin{split}
        \{ 1+\log(1+\eta) \}^{ -(a+b) } &= \int _0^{\infty} \frac{ w^{(a+b)-1} e^{ -w \{ 1+\log(1+\eta) \}  } }{ \Gamma (a+b) }  dw \\
        &= \int _0^{\infty} (1+\eta)^{-w} \frac{ w^{(a+b)-1} e^{ -w } }{ \Gamma (a+b) }  dw.
\end{split}
\end{equation*}
Second, we compute the term of $\log(1+\eta)$ in a similar way: 
\begin{equation*}
\begin{split}
        \{ \log(1+\eta) \}^{ -(1-a) } &= \int _0^{\infty} \frac{ v^{(1-a)-1} e^{ -v \{ \log(1+\eta) \}  } }{ \Gamma (1-a) }  dv \\
        &= \int _0^{\infty} (1+\eta)^{-v} \frac{ v^{-a} }{ \Gamma (1-a) }  dv.
\end{split}
\end{equation*}
Note that $1-a>0$ since we have assumed that $0<a<1$. Last, the term of $(1+\eta)$ is augmented as 
\begin{equation*}
\begin{split}
        (1+\eta) ^{ - (v+w+1) } &= \int _0^{\infty} \frac{ u^{(v+w+1)-1} e^{ - (1+\eta) u } }{ \Gamma (v+w+1) }  du \\
        &= \int _0^{\infty} ue^{-u\eta} \frac{ u^{(v+w)-1} e^{-u} }{ \Gamma (v+w) } \frac{1}{v+w} du.
\end{split}
\end{equation*}
Combining these three expressions, we obtain the desired result.

\subsection*{Full conditional of $(u,v,w)$}

Theorem~1 claims that the RSB distribution is realized as the marginal distribution of the following hierarchical model:
\begin{equation*}
    \begin{split}
        (\eta|u,v,w) &\sim {\rm Ex}(u) \\
        (u|v,w) &\sim {\rm Ga}(v+w,1) \\
        (v,w) &\sim g(v,w).
    \end{split}
\end{equation*}
In this model, the joint density of $(u,v,w)$ conditional on $\eta$ can be computed by backtracking the augmentation made in the proof of Theorem~1:
\begin{equation*}
    \begin{split}
        \pi(u,v,w|\eta) &\propto \pi (u,v,w,\eta) \\
        & = ue^{-u\eta} \ \frac{u^{v+w-1}e^{-u} }{\Gamma (v+w)} \ \frac{1}{B(a,b)\Gamma (1-a)\Gamma (a+b)} \frac{v^{-a}w^{a+b-1}e^{-w} }{ v+w } \\
        &\propto \frac{u^{(v+w+1)-1}e^{-(1+\eta)u} }{\Gamma (v+w+1)} \ v^{-a}w^{a+b-1}e^{-w} \\
        &\propto \frac{u^{(v+w+1)-1}e^{-(1+\eta)u} }{\Gamma (v+w+1) (1+\eta)^{-(v+w+1)} } \ e^{-(v+w+1)\log(1+\eta)} v^{-a}w^{a+b-1}e^{-w} \\
        &\propto {\rm Ga} (u|v+w+1,1+u) \ v^{(1-a)-1} e^{-v\log(1+\eta)} \ w^{a+b-1}e^{-w \{ 1+\log(1+\eta) \} } \\
        &\propto {\rm Ga} (u|v+w+1,1+u) \ {\rm Ga} (v|1-a,\log(1+\eta) ) \ {\rm Ga} (w | a+b,  1+\log(1+\eta) ).
    \end{split}
\end{equation*}
That is, the joint distribution of interest is decomposed as: $\pi(u,v,w|\eta) = \pi (u|v,w,\eta) \pi (v,w|\eta)$, where $\pi (u|v,w,\eta)$ is the gamma distribution. This expression also implies the conditional independence: $\pi( v,w | \eta) = \pi (v|\eta) \pi (w|\eta)$, where $\pi (v|\eta)$ and $\pi (w|\eta)$ are also the gamma distributions. In simulating $(u,v,w)$ from this full conditional, one should
\begin{itemize}
    \item Generate $v \sim {\rm Ga} (v|1-a,\log(1+\eta) )$ and  $w \sim {\rm Ga} (w | a+b,  1+\log(1+\eta) )$ independently. 
    
    \item Then, using the simulated values of $v$ and $w$, generate $u \sim {\rm Ga} (u|v+w+1,1+\eta)$.
\end{itemize}
This is exactly the last step of the Gibbs sampler of Algorithm~1 in Section~3.2.

\subsection*{Moments}

For $\eta|\tilde{v}\sim {\rm Ex}(\tilde{v})$, we have $E[\eta|\tilde{v}] = 1 / \tilde{v}$. In the hierarchical expression of ${\rm SB}(a,b)$, we have $\eta | \tilde{v} \sim {\rm Ga}(b,\tilde{w})$ and $\tilde{w}\sim {\rm Be} (a,1-a)$. Then, the conditional expectation of $\eta $ is 
\begin{equation*}
    E[\eta|\tilde{w}] =  E[ \tilde{v}^{-1}| \tilde{w} ] = \int _0^{\infty} \frac{\tilde{v}^{(b-1)-1}e^{-\tilde{w}\tilde{v}} }{ \Gamma (b) \tilde{w}^{-b} } d\tilde{v} = \begin{cases}
    \infty \ \ \ \ \ \ \mathrm{if} \ b\le 1 \\
    b^{-1} \tilde{w} \ \ \  \mathrm{if} \ b> 1,
    \end{cases}
\end{equation*}
and, for $b>1$, $E[\eta] = a/b$. That is, the first-order moment of the SB distribution exists if and only if $b>1$. 

By contrast, when $\eta \sim {\rm RSB}(a,b)$, we have $\eta|u\sim {\rm Ex}(u)$, $u|\tilde{v} \sim {\rm Ga}(\tilde{v},1)$ and $\tilde{v}$ follows some distribution on $(0,\infty)$. Then, the moment is computed as 
\begin{equation*}
    E[\eta |\tilde{v}] = E[u^{-1}|\tilde{v}] = \int _0^{\infty} \frac{u^{(\tilde{v}-1)-1}e^{-u} }{ \Gamma (\tilde{v}) } d\tilde{v} = \begin{cases}
    \infty \ \ \ \ \ \ \mathrm{if} \ \tilde{v}\le 1 \\
    \tilde{v} \ \ \ \ \ \ \ \mathrm{if} \ \tilde{v}> 1.
    \end{cases}
\end{equation*}
Since the density of $\tilde{v}$ is strictly positive at any $\tilde{v}\le 1$, the expectation of the expression above in terms of $\tilde{v}$ diverges. Thus, the first-order moment of $\eta$ does not exist, for any choice of $(a,b)$. This property reflects the super heavy tail of the RSB density.

%-----------------------------------------------------------%
%       Proof of Theorem 2
%-----------------------------------------------------------%
\section{Proof of Theorem 2} 
Let $G_A (u) = (1 - s) {\rm{Ga}} (u | A, A) + s \pi _{\rm{RSB}} (u; a, b)$, $u \in (0, \infty )$, for $A > 0$. 
Then for any $\la , \De > 0$, by the covariance inequality, we have 
\begin{align}
\{ \Pr( \eta _i<\ep|y_i=0; a, b) \} |_{\la _i = \la } 
&= \lim_{A \to \infty } \frac{ \int_{0}^{\infty } \chi _{(0, \ep )} (u) e^{- \la u} G_A (u) du }{ \int_{0}^{\infty } e^{- \la u} G_A (u) du } \non \\
&\le \lim_{A \to \infty } \frac{ \int_{0}^{\infty } \chi _{(0, \ep )} (u) e^{- \De u} e^{- \la u} G_A (u) du / \int_{0}^{\infty } e^{- \la u} G_A (u) du }{ \int_{0}^{\infty } e^{- \De u} e^{- \la u} G_A (u) du / \int_{0}^{\infty } e^{- \la u} G_A (u) du } \non \\
&= \{ \Pr( \eta _i<\ep|y_i=0; a, b) \} |_{\la _i = \la + \De } \text{.} \non 
\end{align}
Therefore, $\Pr( \eta _i < \ep | y_i=0; a, b)$ is a nondecreasing function of $\la _i$. 
Next, for $0 < a < 1 < a'$, let $a''$ be either $a$ or $a'$. 
Suppose that $s \in (0, 1]$. 
Then, since 
\begin{align}
&\{ \Pr( \eta _i<\ep|y_i=0; a'' , b) \} |_{\la _i = \la } \non \\
&= \lim_{A \to \infty } \frac{ (1 - s) \{ A^A / \Ga (A) \} \int_{0}^{\ep } u^{A - 1} e^{- ( \la + A) u} du + s \int_{0}^{\ep } e^{- \la u} \pi _{\rm{RSB}} (u; a'', b) d u }{ (1 - s) \{ A^A / \Ga (A) \} \int_{0}^{\infty } u^{A - 1} e^{- ( \la + A) u} du + s \int_{0}^{\infty } e^{- \la u} \pi _{\rm{RSB}} (u; a'' , b) du } \non \\
&= \lim_{A \to \infty } \frac{ (1 - s) \{ A^{A - 1 / 2} e^{- A} / \Ga (A) \} \int_{0}^{\ep } \sqrt{A} (u e^{1 - u} )^{A - 1} e^{1 - ( \la + 1) u} du + s \int_{0}^{\ep } e^{- \la u} \pi _{\rm{RSB}} (u; a'' , b) du }{ (1 - s) A^A / ( \la + A)^A + s \int_{0}^{\infty } e^{- \la u} \pi _{\rm{RSB}} (u; a'' , b) du } \non \\
&= \frac{ s \int_{0}^{\ep } e^{- \la u} \pi _{\rm{RSB}} (u; a'' , b) du }{ (1 - s) e^{- \la } + s \int_{0}^{\infty } e^{- \la u} \pi _{\rm{RSB}} (u; a'' , b) du } \non 
\end{align}
for $0 < \ep < 1 / e$ by Stirling's formula and the dominated convergence theorem, it follows that 
\begin{align}
&{(1 - s) e^{- \la } + s \int_{0}^{\infty } e^{- \la u} \pi _{\rm{RSB}} (u; a, b) du \over (1 - s) e^{- \la } + s \int_{0}^{\infty } e^{- \la u} \pi _{\rm{RSB}} (u; a' , b) du} {\Pr( \eta _i<\ep|y_i=0; a, b) \over \Pr( \eta _i<\ep|y_i=0; a' , b)} \Big| _{\la _i = \la } \non \\
&= {\int_{0}^{\ep } e^{- \la u} \pi _{\rm{RSB}} (u; a, b) du \over \int_{0}^{\ep } e^{- \la u} \pi _{\rm{RSB}} (u; a' , b) du} \non \\
&= {B( a' , b) \over B(a, b)} {\int_{0}^{\ep } e^{- \la u} \{ \log (1 + u) \} ^{a - 1} (1 + u)^{- 1} \{ 1 + \log (1 + u) \} ^{- (a + b)} du \over \int_{0}^{\ep } e^{- \la u} \{ \log (1 + u) \} ^{a' - 1} (1 + u)^{- 1} \{ 1 + \log (1 + u) \} ^{- ( a' + b)} du} \non \\
&\ge {B( a' , b) \over B(a, b)} {\int_{0}^{\ep } \{ \log (1 + u) \} ^{a - 1} du / \int_{0}^{\ep } \{ \log (1 + u) \} ^{a' - 1} du \over e^{\la \ep } (1 + \ep ) \{ 1 + \log (1 + \ep ) \} ^{a + b} } \non \\
&\ge {B( a' , b) \over B(a, b)} {\{ \log (1 + \ep ) \} ^{a - a'} \over e^{\la \ep } (1 + \ep ) \{ 1 + \log (1 + \ep ) \} ^{a + b} } \to \infty \non 
\end{align}
as $\ep \to 0$. 
This proves the desired result. 

The proof above is valid even if one replaces $\pi _{\mathrm{RSB}}(u;a,b)$ with $\pi _{\mathrm{SB}}(u;a,b)$. The last inequality becomes 
\begin{align}
{\int_{0}^{\ep } e^{- \la u} \pi _{\rm{SB}} (u; a, b) du \over \int_{0}^{\ep } e^{- \la u} \pi _{\rm{SB}} (u; a' , b) du} &= {B( a' , b) \over B(a, b)} {\int_{0}^{\ep } e^{- \la u} u^{a - 1} (1 + u)^{- (a+b)} du \over \int_{0}^{\ep } e^{- \la u} u^{a' - 1} (1 + u)^{- (a'+b)} du} \non \\
&\ge {B( a' , b) \over B(a, b)} {\int_{0}^{\ep } u ^{a - 1} du / \int_{0}^{\ep } u^{a' - 1} du \over e^{\la \ep } (1 + \ep ) ^{a + b} } \non \\
&\ge {B( a' , b) \over B(a, b)} { \ep ^{a - a'} \over e^{\la \ep } (1 + \ep ) ^{a + b} } \to \infty .
\end{align}

%-----------------------------------------------------------%
%       Useful Lemmas
%-----------------------------------------------------------%
\section{Lemmas}\label{sec:lem}
We here provide some useful lemmas used in the proofs of Theorems \ref{thm:condition} and 3.
Let 
\begin{align}
f_{\pi } (y; z) &= \int_{0}^{\infty } \pi (u) {\rm{Po}} (y | e^z u) du \non \\
&= \int_{0}^{\infty } \pi (u) {( e^z u)^{y} \over y !} \exp (- e^z u) du = \int_{0}^{\infty } {\pi (u / e^z ) \over e^z} {u^{y} \over y !} \exp (- u) du \non 
\end{align}
for $y \in \mathbb{N} _0 = \{ 0, 1, 2, \dotsc \} $ and $z \in \mathbb{R}$ for integrable $\pi ( \cdot ) \colon (0, \infty ) \to (0, \infty )$.

\begin{lem}
\label{lem:pointwise_convergence} 
Suppose that 
\begin{align}
\pi (u) &\sim {1 \over (1 + u)^{1 + \de }} {1 \over \{ 1 + \log (1 + u) \} ^{1 + b}} \non 
\end{align}
as $u \to \infty $ for some $b, \de \in \mathbb{R}$. 
Then 
\begin{align}
{f_{\pi } (y; z) \over f_{\pi } (y; 0)} &\to e^{\de z} \non 
\end{align}
as $y \to \infty $ for all $z \in \mathbb{R}$. 
\end{lem}

\noindent
\begin{proof}
Fix $z \in \mathbb{R}$. 
We have 
\begin{align}
{f_{\pi } (y; z) \over f_{\pi } (y; 0) e^{\de z}} &= \frac{ \displaystyle \int_{0}^{\infty } {\pi (u / e^z ) \over e^{(1 + \de )z}} u^y e^{- u} du }{ \displaystyle \int_{0}^{\infty } \pi (u) u^y e^{- u} du } = \frac{ \displaystyle \int_{0}^{\infty } {\pi (u / e^z ) \over \pi (u) e^{(1 + \de )z}} \pi (u) u^y e^{- u} du }{ \displaystyle \int_{0}^{\infty } \pi (u) u^y e^{- u} du } \text{.} \non 
\end{align}
For any $M > 0$ and $j = 0, 1$, 
\begin{align}
0 &\le \int_{0}^{\infty } \Big\{ {\pi (u / e^z ) \over \pi (u) e^{(1 + \de )z}} \Big\} ^j \pi (u) u^y e^{- u} du / \int_{M}^{\infty } \Big\{ {\pi (u / e^z ) \over \pi (u) e^{(1 + \de )z}} \Big\} ^j \pi (u) u^y e^{- u} du - 1 \non \\
&\le \int_{0}^{M} \Big\{ {\pi (u / e^z ) \over \pi (u) e^{(1 + \de )z}} \Big\} ^j \pi (u) u^y e^{- u} du / \int_{M + 1}^{\infty } \Big\{ {\pi (u / e^z ) \over \pi (u) e^{(1 + \de )z}} \Big\} ^j \pi (u) u^y e^{- u} du \non \\
&\le \Big( {M \over M + 1} \Big) ^y \int_{0}^{M} \Big\{ {\pi (u / e^z ) \over \pi (u) e^{(1 + \de )z}} \Big\} ^j \pi (u) e^{- u} du / \int_{M + 1}^{\infty } \Big\{ {\pi (u / e^z ) \over \pi (u) e^{(1 + \de )z}} \Big\} ^j \pi (u) e^{- u} du \to 0 \non 
\end{align}
as $y \to \infty $. 
For any $y \in \mathbb{N} _0$, 
\begin{align}
0 &\le \Big| \int_{M}^{\infty } {\pi (u / e^z ) \over \pi (u) e^{(1 + \de )z}} \pi (u) u^y e^{- u} du / \int_{M}^{\infty } \pi (u) u^y e^{- u} du - 1 \Big| \non \\
&\le \int_{M}^{\infty } \Big| {\pi (u / e^z ) \over \pi (u) e^{(1 + \de )z}} - 1 \Big| \pi (u) u^y e^{- u} du / \int_{M}^{\infty } \pi (u) u^y e^{- u} du \non \\
&\le \sup_{u \in (M, \infty )} \Big| {\pi (u / e^z ) \over \pi (u) e^{(1 + \de )z}} - 1 \Big| \to 0 \non 
\end{align}
as $M \to \infty $, since 
\begin{align}
\lim_{u \to \infty } {\pi (u / e^z ) \over \pi (u) e^{(1 + \de )z}} &= \lim_{u \to \infty } {(1 + u)^{1 + \de } \over (e^z + u)^{1 + \de }} {\{ 1 + \log (1 + u) \} ^{1 + b} \over \{ 1 + \log (1 + u / e^z ) \} ^{1 + b}} = 1 \text{.} \non 
\end{align}
Thus, 
\begin{align}
{f_{\pi } (y; z) \over f_{\pi } (y; 0) e^{\de z}} \to 1 \non 
\end{align}
as $y \to \infty $. 
\end{proof}

\bigskip
The following result is proved in the proof of Lemma S6 of \cite{hamura2020log}.

\begin{lem}
\label{lem:log} 
For all $u, v \in (0, \infty )$, we have 
\begin{align}
{1 + \log (1 + u) \over 1 + \log (1 + u / v)} &\le 1 + \log (1 + v) \text{.} \non 
\end{align}
\end{lem}

\begin{proof}
We have 
\begin{align}
{1 + \log (1 + u) \over 1 + \log (1 + u / v)} &\le{1 + \log (1 + u / v) + \log (1 + v) \over 1 + \log (1 + u / v)} \le 1 + \log (1 + v) \non 
\end{align}
and the result follows. 
\end{proof}

The following result is a special case of Lemma S4 of \cite{hamura2020log}. 
\begin{lem}
\label{lem:linear_0w} 
Let $m, p \in \mathbb{N}$. 
Let $Z = ( z_1 , \dots , z_m )^{\top}$ be an $m \times p$ matrix of observations (such that any set of its $p$ distinct row vectors is linearly independent). 
Suppose that $m \ge p$. 
Then there exist $R > 0$ and $\de > 0$ (which may depend on $Z$) such that 
\begin{align}
\prod_{i = 1}^{m} {1 \over 1 + | z_{i}^{\top} \be |} \le {1 \over (1 + \de | \be | )^{m - p + 1}} \non 
\end{align}
for all $\be \in \mathbb{R} ^p$ satisfying $| \be | \ge R$. 
\end{lem}

%-----------------------------------------------------------%
%      Conditions for posterior robustness
%-----------------------------------------------------------%
\section{Conditions for posterior robustness}
\label{sec:S_condition} 

We here show that the tail decay of the RSB distribution is essential for the posterior robustness property considered in Theorem~3.
To this end, we consider a more general integrable density which, for some $a > 0$, $b \ge -1$, and $\de \geq 0$, satisfies the following two properties: 
\begin{align}
\pi (u) &\sim u^{a - 1} \quad \text{as $u \to 0$} \text{,} \label{eq:ltail} \\
\pi (u) &\sim {1 \over (1 + u)^{1 + \de }} {1 \over \{ 1 + \log (1 + u) \} ^{1 + b}} \quad \text{as $u \to \infty $} \text{,} \label{eq:rtail} 
\end{align}
where $\pi (u) \sim h(u)$ means that $\pi (u) / h(u) \to C$ for some $C > 0$.
Note that $\delta$ controls the tail decay of the density $\pi(u)$; that is, the tail gets lighter as $\delta$ increases. 
The proposed RSB distribution corresponds to the case with $\delta=0$ and $b>0$ and has the heaviest tail in the above family. The SB distribution is also included in this class of distributions as the case of $\delta > 0$ and $b = -1$. 
Then our framework is summarized as follows: 
\begin{align}
&y_i \sim {\rm{Po}} ( y_i | \exp ( x_{i}^{\top} \be + g_{i}^{\top} \xi _i ) u_i ) \text{,} \quad u_i \sim \pi ( u_i ) \text{,} \quad ( \be , \xi ) \sim \pi _{\be , \xi } ( \be , \xi ) \text{;} \non 
\end{align}
the observation $y_i$ is fixed for $i \in \Kc \subset \{ 1, \dots , n \} $ while $y_i = y_i ( \om ) \to \infty $ as $\om \to \infty $ for $i \in \Lc = \{ 1, \dots , n \} \setminus \Kc $. 

An explicit form of the limiting distribution of $( \be , \xi ) | y$ as $\om \to \infty $ is derived in the following theorem.

\begin{thm}
\label{thm:condition} 
Suppose that the prior distribution of $( \beta , \xi )$ is multivariate normal. 
Then we have 
\begin{align}
p( \be , \xi | y) \to \frac{ p( \be , \xi | y _{\Kc } ) \prod_{i \in \Lc } e^{\de ( x_{i}^{\top} \be + g_{i}^{\top} \xi _i )} }{ \int_{\mathbb{R} ^p} \big\{ p( \be , \xi | y _{\Kc } ) \prod_{i \in \Lc } e^{\de ( x_{i}^{\top} \be + g_{i}^{\top} \xi _i )} \big\} d( \be , \xi ) } \non 
\end{align}
as $\om \to \infty $. 
In particular, $p( \be , \xi | y) \to p( \be , \xi | y _{\Kc } )$ as $\om \to \infty $ if and only if $\de = 0$. 
\end{thm}

Theorem \ref{thm:condition} clearly shows that in order for the posterior $p( \be , \xi | y)$ to be robust against the outliers $y_{\Lc } = y \setminus y_{\Kc }$, it is necessary that 
the error distribution $\pi ( \cdot )$ have an extremely heavy tail with $\de = 0$. 
The RSB mixture distribution has this property, but any polynomial tailed distribution with $\de > 0$, including the SB distribution, fails to have the property. 
The error distribution we propose in the main text satisfies the condition $\de = 0$, and the conclusion of Theorem \ref{thm:condition} holds under our proposed model.

\begin{proof}[Proof of Theorem \ref{thm:condition}]
Since the joint prior distribution of $( \beta , \xi )$ is multivariate normal, without loss of generality, we can assume the model 
\begin{align}
&y_i \sim {\rm{Po}} ( y_i | \exp ( x_{i}^{\top} \be ) u_i ) \text{,} \quad u_i \sim \pi ( u_i ) \text{,} \quad \be \sim \pi _{\be } ( \be ) 
\text{,} \non 
\end{align}
where $\pi _{\be } ( \be )$ is a multivariate normal density. 
Let 
\begin{align}
g(u) &= {u^{a - 1} \over (1 + u)^{a + \de }} {1 \over \{ 1 + \log (1 + u) \} ^{1 + b}} \non 
\end{align}
for $u \in (0, \infty )$. 
Then by assumption there exists $M > 0$ such that 
\begin{align}
(1 / M) g(u) \le \pi (u) \le M g(u) \non 
\end{align}
for all $u \in (0, \infty )$. 
Therefore, for all $u \in (0, \infty )$ and all $z \in \mathbb{R}$, 
\begin{align}
{\pi (u / e^z) \over \pi (u) e^{(1 + \de ) z}} &\le M^2 {g(u / e^z ) \over g(u) e^{(1 + \de ) z}} = M^2 \Big( {1 + u \over e^z + u} \Big) ^{a + \de } \Big\{ {1 + \log (1 + u) \over 1 + \log (1 + u / e^z ) } \Big\} ^{1 + b} \non \\
&\le M^2 ( \max \{ 1, e^{- z} \} )^{a + \de } ( \max \{ 1, e^z \} )^{1 + b} \le M^2 \{ e^{- (a + \de ) z} + e^{(1 + b) z} \} \non 
\end{align}
since $\log (1 + u / \max \{ 1, e^z \} ) \ge (1 / \max \{ 1, e^z \} ) \log (1 + u)$. 
It follows that for all $y \in \mathbb{N} _0$ and all $z \in \mathbb{R}$, 
\begin{align}
{f_{\pi } (y; z) \over f_{\pi } (y; 0)} &= \int_{0}^{\infty } {\pi (u / e^z ) \over \pi (u) e^z} \pi (u) u^y e^{- u} du / \int_{0}^{\infty } \pi (u) u^y e^{- u} du \le M^2 \{ e^{- a z} + e^{(1 + b + \de ) z} \} \text{.} \non 
\end{align}
Also, 
\begin{align}
p( y_{\Kc } | \be ) &= \prod_{i \in \Kc } f_{\pi } ( y_i ; x_{i}^{\top} \be ) \le \Big\{ \sup_{(y, \la ) \in \mathbb{N} _0 \times (0, \infty )} {\rm{Po}} (y | \la ) \Big\} ^{| \Kc |} < \infty \text{.} \non 
\end{align}
Thus, by Lemma \ref{lem:pointwise_convergence} and the dominated convergence theorem, 
\begin{align}
{p( \be | y) \over p( \be | y _{\Kc } )} &= {p( y_{\Kc } ) \over p(y)} \prod_{i \in \Lc } f_{\pi } ( y_i ; x_{i}^{\top} \be ) = \frac{ p( y_{\Kc } ) \prod_{i \in \Lc } \{ f_{\pi } ( y_i ; x_{i}^{\top} \be ) / f_{\pi } ( y_i ; 0) \} }{ \int_{\mathbb{R} ^p} \big[ \pi _{\be } ( \be ) p( y_{\Kc } | \be ) \prod_{i \in \Lc } \{ f_{\pi } ( y_i ; x_{i}^{\top} \be ) / f_{\pi } ( y_i ; 0) \} \big] d\be } \non \\
&\to \frac{ \prod_{i \in \Lc } e^{\de x_{i}^{\top} \be } }{ \int_{\mathbb{R} ^p} \big\{ p( \be | y _{\Kc } ) \prod_{i \in \Lc } e^{\de x_{i}^{\top} \be} \big\} d\be } \non 
\end{align}
as $\om \to \infty $. 
This completes the proof. 
\end{proof}

%-----------------------------------------------------------%
%       Proof of Theorem 3
%-----------------------------------------------------------%
\section{Proof of Theorem~3} 
Here we prove Theorem~3. 

\begin{proof}[Proof of Theorem~3]
We have 
\begin{align}
{p( \be , \xi | y) \over p( \be , \xi | y_{\Kc } )} &= \frac{ \displaystyle \Big[ \int_{\mathbb{R} ^p \times \mathbb{R} ^q} \pi _{\be , \xi } ( \be , \xi ) \Big\{ \prod_{i \in \Kc } f_{\pi } ( y_i ; x_{i}^{\top} \be + g_{i}^{\top} \xi ) \Big\} d( \be , \xi ) \Big] \prod_{i \in L} {f_{\pi } ( y_i ; x_{i}^{\top} \be + g_{i}^{\top} \xi ) \over f_{\pi } ( y_i ; 0)} }{ \displaystyle \int_{\mathbb{R} ^p \times \mathbb{R} ^q} \pi _{\be , \xi } ( \be , \xi ) \Big\{ \prod_{i \in \Kc } f_{\pi } ( y_i ; x_{i}^{\top} \be + g_{i}^{\top} \xi ) \Big\} \Big\{ \prod_{i \in \Lc } {f_{\pi } ( y_i ; x_{i}^{\top} \be + g_{i}^{\top} \xi ) \over f_{\pi } ( y_i ; 0)} \Big\} d( \be , \xi ) } \text{.} \non 
\end{align}
Then by Lemmas \ref{lem:DCT} and \ref{lem:pointwise_convergence} and the dominated convergence theorem, 
\begin{align}
{p( \be , \xi | y) \over p( \be , \xi | y_{\Kc } )} &\to 1 \non 
\end{align}
as $\om \to \infty $. 
This proves Theorem~3. 
\end{proof}

\begin{lem}
\label{lem:DCT} 
Under the assumptions of Theorem~3, there exists an integrable function
$\overline{h} ( \be, \xi )$ of $( \be, \xi )$ which does not depend on $\om $ such that 
\begin{align}
\pi _{\be , \xi } ( \be , \xi ) \Big\{ \prod_{i \in \Kc } f_{\pi } ( y_i ; x_{i}^{\top} \be + g_{i}^{\top} \xi ) \Big\} \prod_{i \in \Lc } {f_{\pi } ( y_i ; x_{i}^{\top} \be + g_{i}^{\top} \xi ) \over f_{\pi } ( y_i ; 0)} &\le \overline{h} ( \be, \xi ) \text{.} \non 
\end{align}
\end{lem}

\begin{proof}
Let 
\begin{align}
g(u) &= {u^{a - 1} \over (1 + u)^a} {1 \over \{ 1 + \log (1 + u) \} ^{1 + b}} \non 
\end{align}
for $u \in (0, \infty )$. 
Then by assumption there exists $M_1 > 0$ such that 
\begin{align}
(1 / M_1 ) g(u) \le \pi (u) \le M_1 g(u) \non 
\end{align}
for all $u \in (0, \infty )$. 

Fix $z \in \mathbb{R}$. 
Then for all $y \in \mathbb{N}$, 
\begin{align}
{f_{\pi } (y; z) \over f_{\pi } (y; 0)} &= {1 \over e^z} \frac{ \int_{0}^{\infty } \pi (u / e^z ) u^y e^{- u} du }{ \int_{0}^{\infty } \pi (u) u^y e^{- u} du } \le {{M_1}^2 \over e^z} \frac{ \int_{0}^{\infty } g(u / e^z ) u^y e^{- u} du }{ \int_{0}^{\infty } g(u) u^y e^{- u} du } \text{.} \non 
\end{align}
Let $M_2 > \sup_{u \in (0, \infty )} \{ (1 + u)^a / (1 + u^a ) \} $. 
Then 
\begin{align}
\int_{0}^{\infty } {g(u / e^z ) \over e^z} u^y e^{- u} du &= \int_{0}^{\infty } {g(u / e^z ) \over g(u) e^z} g(u) u^y e^{- u} du \non \\
&= \int_{0}^{\infty } \Big( {1 + u \over e^z + u} \Big) ^a \Big\{ {1 + \log (1 + u) \over 1 + \log (1 + u / e^z ) } \Big\} ^{1 + b} g(u) u^y e^{- u} du \non \\
&\le M_2 \int_{0}^{\infty } {1 + u^a \over u^a} \Big\{ {1 + \log (1 + u) \over 1 + \log (1 + u / e^z ) } \Big\} ^{1 + b} g(u) u^y e^{- u} du \non \\
&\le M_2 \{ 1 + \log (1 + e^z ) \} ^{1 + b} \int_{0}^{\infty } g(u) ( u^{y - a} + u^y ) e^{- u} du \non 
\end{align}
for all $y \in \mathbb{N}$, where the second inequality follows from Lemma \ref{lem:log}. 
Since $g(u) u^{1 - a} e^{- u}$ is an integrable function of $u \in (0, \infty )$, 
\begin{align}
0 &\le \frac{ \int_{0}^{\infty } g(u) u^{y - a} e^{- u} du }{ \int_{1}^{\infty } g(u) u^{y - a} e^{- u} du } - 1 \le \frac{ \int_{0}^{1} g(u) u^{y - a} e^{- u} du }{ \int_{2}^{\infty } g(u) u^{y - a} e^{- u} du } \le {1 \over 2^{y - 1}} \frac{ \int_{0}^{1} g(u) u^{1 - a} e^{- u} du }{ \int_{2}^{\infty } g(u) u^{1 - a} e^{- u} du } \to 0 \non 
\end{align}
as $y \to \infty $. 
Therefore, there exists $M_3 > 0$ such that 
\begin{align}
\frac{ \int_{0}^{\infty } g(u) u^{y - a} e^{- u} du }{ \int_{1}^{\infty } g(u) u^{y - a} e^{- u} du } &\le M_3 \non 
\end{align}
for all $y \in \mathbb{N}$. 
Thus, 
\begin{align}
{f_{\pi } (y; z) \over f_{\pi } (y; 0)} &\le {M_1}^2 M_2 \{ 1 + \log (1 + e^z ) \} ^{1 + b} \Big\{ \frac{ \int_{0}^{\infty } g(u) u^{y - a} e^{- u} du }{ \int_{0}^{\infty } g(u) u^y e^{- u} du } + 1 \Big\} \non \\
&\le {M_1}^2 M_2 \{ 1 + \log (1 + e^{|z|} ) \} ^{1 + b} \Big\{ \frac{ \int_{1}^{\infty } g(u) u^{y - a} e^{- u} du }{ \int_{1}^{\infty } g(u) u^y e^{- u} du } \frac{ \int_{0}^{\infty } g(u) u^{y - a} e^{- u} du }{ \int_{1}^{\infty } g(u) u^{y - a} e^{- u} du } + 1 \Big\} \non \\
&\le {M_1}^2 M_2 \{ 1 + \log (1 + e^{|z|} ) \} ^{1 + b} ( M_3 + 1) \non \\
&\le {M_1}^2 M_2 ( M_3 + 1) {M_4}^{1 + b} (1 + |z|)^{1 + b} \non 
\end{align}
for all $y \in \mathbb{N}$ for some $M_4 > \sup_{r \in [0, \infty )} \{ 1 + \log (1 + e^r ) \} / (1 + r)$. 
It follows that 
\begin{align}
{f_{\pi } ( y_i ; x_{i}^{\top} \be + g_{i}^{\top} \xi ) \over f_{\pi } ( y_i ; 0)} &\le {M_1}^2 M_2 ( M_3 + 1) {M_4}^{1 + b} (1 + | x_{i}^{\top} \be + g_{i}^{\top} \xi |)^{1 + b}  \non \\
&\le {M_1}^2 M_2 ( M_3 + 1) {M_4}^{1 + b} (1 + | x_{i}^{\top} \be |)^{1 + b} (1 + | g_{i}^{\top} \xi |)^{1 + b} \label{lDCTp1} 
\end{align}
for all $i \in \Lc $ (for sufficiently large $\om $). 

Fix $i \in \Kc _{+}$. 
Then 
\begin{align}
f_{\pi } ( y_i ; x_{i}^{\top} \be + g_{i}^{\top} \xi ) &= \int_{0}^{\infty } {\pi (u / e^{x_{i}^{\top} \be + g_{i}^{\top} \xi } ) \over e^{x_{i}^{\top} \be + g_{i}^{\top} \xi }} {u^{y_i} \over y_i !} e^{- u} du \le M_1 \int_{0}^{\infty } {g(u / e^{x_{i}^{\top} \be + g_{i}^{\top} \xi } ) \over e^{x_{i}^{\top} \be + g_{i}^{\top} \xi }} {u^{y_i} \over y_i !} e^{- u} du \non \\
&= M_1 \int_{0}^{\infty } {u^{y_i + a - 1} e^{- u} / ( y_i !) \over ( e^{x_{i}^{\top} \be + g_{i}^{\top} \xi } + u)^a} {1 \over \{ 1 + \log (1 + u / e^{x_{i}^{\top} \be + g_{i}^{\top} \xi } ) \} ^{1 + b}} du \non \\
&\le M_1 \int_{0}^{\infty } {u^{y_i + a - 1} e^{- u} / ( y_i !) \over ( e^{x_{i}^{\top} \be + g_{i}^{\top} \xi } + u)^a} {\{ 1 + \log (1 + 1 / u) \} ^{1 + b} \over \{ 1 + \log (1 + 1 / e^{x_{i}^{\top} \be + g_{i}^{\top} \xi } ) \} ^{1 + b}} du \text{,} \non 
\end{align}
where the second inequality follows from Lemma \ref{lem:log}. 
Let $0 < \ep < 1$ and let $M_5 > \sup_{u \in (0, \infty )} \{ u^{\ep } / (1 + u^{\ep } ) \} \{ 1 + \log (1 + 1 / u) \} ^{1 + b}$. 
Then 
\begin{align}
f_{\pi } ( y_i ; x_{i}^{\top} \be + g_{i}^{\top} \xi ) &\le M_1 M_5 \int_{0}^{\infty } {u^{y_i + a - 1} e^{- u} / ( y_i !) \over ( e^{x_{i}^{\top} \be + g_{i}^{\top} \xi } + u)^a} {(1 + u^{\ep } ) / u^{\ep } \over \{ 1 + \log (1 + 1 / e^{x_{i}^{\top} \be + g_{i}^{\top} \xi } ) \} ^{1 + b}} du \non \\
&\le M_1 M_5 e^{a | g_{i}^{\top} \xi |} \int_{0}^{\infty } {u^{y_i + a - 1} e^{- u} / ( y_i !) \over ( e^{x_{i}^{\top} \be } + u)^a} {(1 + u^{\ep } ) / u^{\ep } \over \{ 1 + \log (1 + e^{- | g_{i}^{\top} \xi |} / e^{x_{i}^{\top} \be } ) \} ^{1 + b}} du \non \\
&\le {M_1 M_5 e^{a | g_{i}^{\top} \xi |} (1 + e^{| g_{i}^{\top} \xi |} )^{1 + b} \over \{ 1 + \log (1 + 1 / e^{x_{i}^{\top} \be } ) \} ^{1 + b}} \int_{0}^{\infty } {u^{y_i + a - 1} e^{- u} / ( y_i !) \over ( e^{x_{i}^{\top} \be } + u)^a} \Big( {1 \over u^{\ep }} + 1 \Big) du \text{,} \non 
\end{align}
where the third inequality follows from Lemma \ref{lem:log}. 
Now, if $e^{x_{i}^{\top} \be } \ge 1$, then 
\begin{align}
\int_{0}^{\infty } {u^{y_i + a - 1} e^{- u} / ( y_i !) \over ( e^{x_{i}^{\top} \be } + u)^a} \Big( {1 \over u^{\ep }} + 1 \Big) du &\le e^{- a x_{i}^{\top} \be } \int_{0}^{\infty } {u^{y_i + a - 1} e^{- u} \over y_i !} \Big( {1 \over u^{\ep }} + 1 \Big) du \le M_6 e^{- a x_{i}^{\top} \be } \non 
\end{align}
for some $M_6 > \max_{i' \in \Kc _{+}} \{ \Ga ( y_{i'} + a - \ep ) + \Ga ( y_{i'} + a) \} / ( y_{i'} !)$. 
On the other hand, if $e^{x_{i}^{\top} \be } \le 1$, then 
\begin{align}
\int_{0}^{\infty } {u^{y_i + a - 1} e^{- u} / ( y_i !) \over ( e^{x_{i}^{\top} \be } + u)^a} \Big( {1 \over u^{\ep }} + 1 \Big) du &\le \int_{0}^{\infty } {u^{y_i - 1} e^{- u} \over y_i !} \Big( {1 \over u^{\ep }} + 1 \Big) du \le M_7 \non 
\end{align}
for some $M_7 > \max_{i' \in \Kc _{+}} \{ \Ga ( y_{i'} - \ep ) + \Ga ( y_{i'} ) \} / ( y_{i'} !)$. 
Therefore, letting $M_8 > \sup_{r \ge 0} (1 + r)^{1 + b} e^{- a r}$ and $M_9 > \sup_{r \ge 0} (1 + r)^{1 + b} / \{ 1 + \log (1 + e^r ) \} ^{1 + b}$, we have 
\begin{align}
{f_{\pi } ( y_i ; x_{i}^{\top} \be + g_{i}^{\top} \xi ) \over M_1 M_5 e^{a | g_{i}^{\top} \xi |} (1 + e^{| g_{i}^{\top} \xi |} )^{1 + b}} &\le \begin{cases} \displaystyle {M_6 e^{- a x_{i}^{\top} \be } \over \{ 1 + \log (1 + 1 / e^{x_{i}^{\top} \be } ) \} ^{1 + b}} \le {M_6 M_8 \over (1 + | x_{i}^{\top} \be |)^{1 + b}} \text{,} & \text{if $e^{x_{i}^{\top} \be } \ge 1$} \text{,} \\ \displaystyle {M_7 \over \{ 1 + \log (1 + 1 / e^{x_{i}^{\top} \be } ) \} ^{1 + b}} \le {M_7 M_9 \over (1 + | x_{i}^{\top} \be |)^{1 + b}} \text{,} & \text{if $e^{x_{i}^{\top} \be } \le 1$} \text{.} \end{cases} \non 
\end{align}
Thus, 
\begin{align}
f_{\pi } ( y_i ; x_{i}^{\top} \be + g_{i}^{\top} \xi ) &\le M_{10} e^{a | g_{i}^{\top} \xi |} (1 + e^{| g_{i}^{\top} \xi |} )^{1 + b} {1 \over (1 + | x_{i}^{\top} \be |)^{1 + b}} \label{lDCTp2} 
\end{align}
for some $M_{10} > 0$. 

Finally, note that 
\begin{align}
f_{\pi } ( y_i ; x_{i}^{\top} \be + g_{i}^{\top} \xi ) &= \int_{0}^{\infty } \pi (u) \exp (- e^{x_{i}^{\top} \be + g_{i}^{\top} \xi } u) du \le \int_{0}^{\infty } \pi (u) du = 1 \label{lDCTp3} 
\end{align}
if $i \in \Kc _0 = \{ i \in \Kc | y_i = 0 \} $. 
Then, combining (\ref{lDCTp1}), (\ref{lDCTp2}), and (\ref{lDCTp3}), we obtain 
\begin{align}
&\pi _{\be , \xi } ( \be , \xi ) \Big\{ \prod_{i \in \Kc } f_{\pi } ( y_i ; x_{i}^{\top} \be + g_{i}^{\top} \xi ) \Big\} \prod_{i \in \Lc } {f_{\pi } ( y_i ; x_{i}^{\top} \be + g_{i}^{\top} \xi ) \over f_{\pi } ( y_i ; 0)} \non \\
&\le \pi _{\be , \xi } ( \be , \xi ) \Big\{ \prod_{i \in \Kc _{+}} f_{\pi } ( y_i ; x_{i}^{\top} \be + g_{i}^{\top} \xi ) \Big\} \prod_{i \in \Lc } {f_{\pi } ( y_i ; x_{i}^{\top} \be + g_{i}^{\top} \xi ) \over f_{\pi } ( y_i ; 0)} \non \\
&\le \pi _{\be , \xi } ( \be , \xi ) \Big[ \prod_{i \in \Kc _{+}} \Big\{ M_{10} e^{a | g_{i}^{\top} \xi |} (1 + e^{| g_{i}^{\top} \xi |} )^{1 + b} {1 \over (1 + | x_{i}^{\top} \be |)^{1 + b}} \Big\} \Big] \non \\
&\quad \times \prod_{i \in \Lc } \{ {M_1}^2 M_2 ( M_3 + 1) {M_4}^{1 + b} (1 + | x_{i}^{\top} \be |)^{1 + b} (1 + | g_{i}^{\top} \xi |)^{1 + b} \} \text{.} \non 
\end{align}
Since $| \Kc _{+} | \ge p$ by assumption, by Lemma \ref{lem:linear_0w} there exist $R > 0$ and $\al > 0$ such that 
\begin{align}
\prod_{i \in \Kc _{+}} {1 \over 1 + | x_{i}^{\top} \tilde{\be } |} &\le {1 \over (1 + \al | \tilde{\be } |)^{| \Kc _{+} | - p + 1}} \non 
\end{align}
for all $\tilde{\be } \in \mathbb{R} ^p$ satisfying $| \tilde{\be } | \ge R$. 
Thus, 
\begin{align}
&{\pi _{\be , \xi } ( \be , \xi ) \over {M_{10}}^{| \Kc _{+} |} \{ {M_1}^2 M_2 ( M_3 + 1) {M_4}^{1 + b} \} ^{| \Lc |}} \Big\{ \prod_{i \in \Kc } f_{\pi } ( y_i ; x_{i}^{\top} \be + g_{i}^{\top} \xi ) \Big\} \prod_{i \in \Lc } {f_{\pi } ( y_i ; x_{i}^{\top} \be + g_{i}^{\top} \xi ) \over f_{\pi } ( y_i ; 0)} \non \\
&\le \pi _{\be } ( \be | \xi ) \Big\{ \prod_{i \in \Kc _{+}} {1 \over (1 + | x_{i}^{\top} \be |)^{1 + b}} \Big\} \Big\{ \prod_{i \in \Lc } (1 + | x_{i}^{\top} \be |)^{1 + b} \Big\} \non \\
&\quad \times \pi _{\xi } ( \xi ) \Big[ \prod_{i \in \Kc _{+}} \{ e^{a | g_{i}^{\top} \xi |} (1 + e^{| g_{i}^{\top} \xi |} )^{1 + b} \} \Big] \prod_{i \in \Lc } (1 + | g_{i}^{\top} \xi |)^{1 + b} \non \\
&\le \pi _{\be } ( \be | \xi ) \Big[ \sup_{| \tilde{\be } | \le R} {\prod_{i \in \Lc } (1 + | x_{i}^{\top} \tilde{\be } |)^{1 + b} \over \prod_{i \in \Kc _{+}} (1 + | x_{i}^{\top} \tilde{\be } |)^{1 + b}} + \Big\{ {\prod_{i \in \Lc } (1 + | x_{i} | |\be |) \over (1 + \al | \be |)^{| \Kc _{+} | - p + 1}} \Big\} ^{1 + b} \Big] \label{lDCTp4} \\
&\quad \times \pi _{\xi } ( \xi ) \prod_{i = 1}^{n} \{ e^{a | g_{i} | | \xi |} (1 + e^{| g_{i} | | \xi |} )^{1 + b} \} \text{,} \non 
\end{align}
where the prior density of $\xi $, $\pi _{\xi } ( \xi )$, is multivariate normal. 
The right-hand side of the above inequality is an integrable function of $( \be , \xi )$ by assumption. 
This completes the proof. 
\end{proof}

\begin{remark}
\label{remark:thm3} 
The proof of the posterior robustness under an improper prior for $\beta$ is identical until (\ref{lDCTp4}). If the assumption $| \mathcal{K} _{+} | - | \mathcal{L} | - p + 1 \ge 0$ is replaced by $| \mathcal{K} _{+} | - | \mathcal{L} | - p + 1 > p / (1 + b)$, then the conclusion of Theorem~3 holds even if we use a bounded improper prior for $\be | \xi $. 
\end{remark}

%--------------------------------------------------%
%          Supp: Proof of Proposition 1
%--------------------------------------------------%
\section{Proof of Proposition~1}
\label{sec:supp-proof-prop}
Here, we provide a sufficient condition under which posterior moments of $(\beta ,\xi)$ are finite. Proposition~1 is obtained as a corollary of the following statement.

\begin{prp}
\label{prp:supp-posterior_moments} 
Let $\varphi _1 ( \be )$ and $\varphi _2 ( \xi )$ be functions of $\be $ and $\xi $, respectively. 
Let $m = \#\{ i \in \{ 1, \dots , n \} | y_i \ge 1 \} $. 
Assume that $m \ge p$. 
Suppose that the prior distribution of $\xi $ is multivariate normal. 
Suppose that $\varphi _1 ( \be ) / (1 + | \be |)^{(m - p + 1) (1 + b)}$ and $\varphi _2 ( \xi ) / e^{| \xi |}$ are bounded. 
Then $\varphi _1 ( \be ) \varphi _2 ( \xi )$ has finite posterior expectation. 
\end{prp}

\begin{proof}[Proof of Proposition \ref{prp:supp-posterior_moments}.]
Fix $i \in \{ 1, \dots , n \} $. 
We have 
\begin{align}
p( y_i | \be , \xi ) &= (1 - s) {e^{y_i ( x_{i}^{\top} \be + g_{i}^{\top} \xi )} \over y_i !} \exp ( - e^{x_{i}^{\top} \be + g_{i}^{\top} \xi } ) + s f_{\pi } ( y_i ; x_{i}^{\top} \be + g_{i}^{\top} \xi ) \text{.} \non 
\end{align}
If $y_i = 0$, then 
\begin{align}
p( y_i | \be , \xi ) &\le 1 \text{.} \non 
\end{align}
Next, suppose that $y_i \ge 1$. 
Then, by (\ref{lDCTp2}), 
\begin{align}
f_{\pi } ( y_i ; x_{i}^{\top} \be + g_{i}^{\top} \xi ) &\le M_1 e^{(a + 1 + b) | g_{i}^{\top} \xi |} {1 \over (1 + | x_{i}^{\top} \be |)^{1 + b}} \non 
\end{align}
for some $M_1 > 0$. 
On the other hand, 
\begin{align}
&{e^{y_i ( x_{i}^{\top} \be + g_{i}^{\top} \xi )} \over y_i !} \exp ( - e^{x_{i}^{\top} \be + g_{i}^{\top} \xi } ) \non \\
&\le {1 \over y_i !} {e^{y_i ( x_{i}^{\top} \be + g_{i}^{\top} \xi )} \over \{ 1 + e^{x_{i}^{\top} \be + g_{i}^{\top} \xi } / ( y_i + 1) \} ^{y_i + 1}} = {1 \over y_i !} {e^{( y_i + 1) | g_{i}^{\top} \xi |} e^{y_i g_{i}^{\top} \xi } e^{y_i x_{i}^{\top} \be } \over \{ e^{| g_{i}^{\top} \xi |} + e^{| g_{i}^{\top} \xi |} e^{g_{i}^{\top} \xi } e^{x_{i}^{\top} \be } / ( y_i + 1) \} ^{y_i + 1}} \non \\
&\le {1 \over y_i !} {e^{(2 y_i + 1) | g_{i}^{\top} \xi |} e^{y_i x_{i}^{\top} \be } \over \{ 1 + e^{x_{i}^{\top} \be } / ( y_i + 1) \} ^{y_i + 1}} \le {e^{(2 y_i + 1) | g_{i}^{\top} \xi |} \over y_i !} \times \begin{cases} ( y_i + 1)^{y_i + 1} e^{- | x_{i}^{\top} \be |} \text{,} & \text{if $x_{i}^{\top} \be \ge 0$} \text{,} \\ e^{- y_i | x_{i}^{\top} \be |} \text{,} & \text{if $x_{i}^{\top} \be < 0$} \text{.} \end{cases} \non 
\end{align}
Therefore, since $e^{- y_i | x_{i}^{\top} \be |} \le e^{- | x_{i}^{\top} \be |} \le [ \sup_{r \ge 0} \{ (1 + r)^{1 + b} e^{- r} \} ] / (1 + | x_{i}^{\top} \be |)^{1 + b}$, 
\begin{align}
p( y_i | \be , \xi ) &\le M_2 e^{M_3 | g_{i}^{\top} \xi |} {1 \over (1 + | x_{i}^{\top} \be |)^{1 + b}} \non 
\end{align}
for some $M_2 , M_3 > 0$. 
Thus, since $i \in \{ 1, \dots , n \} $ is arbitrary, 
\begin{align}
\pi ( \be , \xi | y ) &\propto \pi ( \be , \xi ) \prod_{i = 1}^{n} p( y_i | \be , \xi ) \non \\
&\le M_4 \pi ( \be , \xi ) e^{M_5 | \xi |} \Big( \prod_{i \in \{ i' \in \{ 1, \dots , n \} | y_i \ge 1 \} } {1 \over 1 + | x_{i}^{\top} \be |} \Big) ^{1 + b} \non 
\end{align}
for some $M_4 , M_5 > 0$, and the result follows from Lemma \ref{lem:linear_0w}. 
\end{proof}

Proposition~1 immediately follows from this result. When $k \in \mathbb{N}$, $\varphi _1 ( \be ) = | \be |^k$ and $\varphi _2 ( \xi ) = 1$, each component of $\be $ has a finite posterior $k$-th moment if $k<(m-p+1)(1+b)$.
This condition is satisfied when $m>p+k-1$, that is, the number of nonzero observations is larger than $p+k-1$.   

\begin{remark}
\label{remark:prop}
The posterior moments under an improper prior for $\beta$ can be evaluated in a similar way. It follows from the above proof that when we use a bounded improper prior for $\be | \xi $, the function $\varphi _1 ( \be ) \varphi _2 ( \xi )$ has finite posterior expectation if the assumption about $\varphi _1 ( \be )$ is replaced by ``$\{ \varphi _1 ( \be ) \} ^j / (1 + | \be |)^{(m - p + 1) (1 + b) - (p + \ep )}$ is bounded for $j = 0, 1$ for some $\ep > 0$.''
\end{remark}

%--------------------------------------------------%
%          Supp: Posterior computation 
%--------------------------------------------------%
\section{Details of posterior computation algorithms in our illustrations}
\label{sec:supp-pos}

%   count regression 
\subsection{Ordinal count regression}
\label{supp:count-reg}

The sampling algorithm for $\eta_i$ is given in Algorithm 1.
We use a multivariate normal prior, $\beta\sim N(A_{\beta}, B_{\beta})$. 
The full conditional distribution of the regression coefficients $\beta$ is proportional to 
$$
\phi(\beta; A_{\beta}, B_{\beta})\exp\left\{\sum_{i=1}^n \left(y_ix_i^{\top}\beta-\eta_i e^{x_i^{\top}\beta}\right)\right\}.
$$
We adopt the independent Metropolis-Hastings (MH) algorithm to generate $\beta$ from its full conditional distribution. 
First, we consider an approximate of the likelihood part.
Let $\beh$ be the maximizer of $\ell(\beta)\equiv \sum_{i=1}^n \left(y_ix_i^{\top}\beta-\eta_i e^{x_i^{\top}\beta}\right)$, and define the Hessian matrix as
$$
\Sih^{-1}\equiv \frac{\partial^2}{\partial\beta\partial\beta^\top}\ell(\beta)|_{\beta=\beh}=\sum _{i=1}^m \eta_i e^{ x_i^\top \beh } x_ix_i^\top. 
$$
The independent MH algorithm generates a proposal from the approximate full conditional distribution obtained by replacing the likelihood part with $N(\beh, \Sih)$.
Let $\beta^{\rm old}$ be the current value and $\beta^{\rm new}$ the proposal generated from the approximated full conditional distribution $N(\At_{\beta}, \Bt_{\beta})$, where
$$
\Bt_{\beta} = ( \Sih^{-1} + B_{\beta}^{-1} )^{-1}, \ \ \ \ \ \ \At_{\beta} = \Bt_{\beta} ( \Sih^{-1}\beh + B_{\beta}^{-1}A_{\beta}).
$$
Then, we accept $\beta^{\rm new}$ with probability 
$$
\min \left\{ \ 1, \ \prod _{i=1}^n \frac{ {\rm Po}(y_i| \eta_i \exp(x_i^{\top}\beta^{\rm new})) \phi(\beh | \beta^{\rm old} , \Sih ) }{{\rm Po}(y_i| \eta_i \exp(x_i^{\top}\beta^{\rm old})) \phi(\beh | \beta^{\rm new} , \Sih )} \ \right\}.
$$

%   trend filtering
\subsection{Locally adaptive smoothing for count}
\label{supp:smoothing}

In what follows, we use $D$ instead of $D^k$ for notational simplicity. 
The sampling algorithm for $\eta_i$ is given in Algorithm 1.
Using the approximated Poisson likelihood model presented in Section 3.2, the sampling algorithm for $\theta=(\theta_1,\ldots,\theta_n)^\top$ and $\tau^2$ is as follows:
\begin{itemize}
\item[-](Sampling from $\theta$) \ \ 
The full conditional distribution of $\theta$ is $N(\At_{\theta}, \Bt_{\theta})$, where $\Bt_{\theta}=(\Omega+DWD^\top)^{-1}$ and $\At_{\theta}=\Bt_{\theta}\{\kappa-\Omega(\log\eta-1_n\log \delta )\}$ with $\Omega={\rm diag}(\omega_1,\ldots,\omega_n)$, $W={\rm diag}(1/\lambda_1\tau^2,\ldots,1/\lambda_{n-K}\tau^2)$, $\kappa=(\kappa_1,\ldots,\kappa_n)$ and $\eta=(\eta_1,\ldots,\eta_n)$.
Here $\kappa_i=(y_i-\delta)/2$ and $\omega_i$ is the latent variables. 

\item[-](Sampling from $\omega_i$) \ \ 
The full conditional distribution of $\omega_i$ is $PG(y_i+\delta, \theta_i+\log\eta_i-\log\delta)$.

\item[-](Sampling from latent parameters in the horseshoe)\ \ 
The full conditional distributions of $\lambda_j, \psi_j$, $\tau^2$ and $\gamma$ as follows: 
\begin{align*}
&\lambda_j|\cdot \sim {\rm IG}\left(1, \frac1{\psi_j}+\frac{(D\theta_j)^2}{2\tau^2}\right), \ \ \ \ 
\psi_j|\cdot \sim {\rm IG}\left(1, 1+\frac1{\lambda_j}\right), \ \ \ \ j=1,\ldots,n-k,\\
& \tau^2|\cdot \sim {\rm IG}\left(\frac{n-k+1}{2}, \frac1{\gamma}+\frac12\sum_{j=1}^{n-k}\frac{(D\theta_j)^2}{\lambda_j}\right), \ \ \ \ \gamma|\cdot \sim {\rm IG}\left(1, 1+\frac{1}{\tau^2}\right).
\end{align*}
\end{itemize}

\subsection{Negative binomial regression}

Our model and MCMC algorithm can cover the negative binomial regression by setting $\gamma _i \sim \mathrm{Ga}(\nu , 1)$. The shape parameter, $\nu$, is estimated in our numerical examples with prior $\nu \sim \mathrm{Unif}(0,10^5)$ (in Section~4.1) or $\nu \sim \mathrm{Ga}(1,1)$ (in Sections~4.2 and 4.3), using the random-walk Metropolis Hastings step and the augmentation method \citep{zhou2013negative}, respectively.

%   count spatial regression 
\subsection{Hierarchical count regression with spatial effects}
\label{supp:count-sp}

Prior distributions of $\beta$, $\tau^2$ and $h$ are as follows: 
$$
\beta\sim N(A_{\beta}, B_{\beta}), \ \ \ \ 
\tau\sim {\rm Ga}(A_{\tau}, B_{\tau}), \ \ \ \
h\sim U(\underline{h}, \overline{h}).
$$
The sampling algorithm for $\eta_i$ is given in Algorithm 1.
Using the approximated Poisson likelihood model given in Section 3.2, the full conditional distributions of $\beta$, $\xi$, $\tau^2$, $h$ and latent variable $\omega_i$ are proportional to 
\begin{align*}
&\pi(\beta)\pi(\tau)\pi(h)|H_h|^{-1/2}\exp\bigg[-\frac12\sum_{i=1}^n\omega_i\big\{x_i^\top\beta+D_h(s_i)^\top\mu_{\xi}+a_i+\log(\eta_i/\delta)\big\}^2 
\\
& \ \ \ \ \ \ 
-\frac12\tau \mu_{\xi}^{\top}H_h^{-1}\mu_{\xi}
+\sum_{i=1}^n \kappa_i\big\{x_i^\top\beta+D_h(s_i)^\top\mu_{\xi}\big\}\bigg]\prod_{i=1}^n p_{PG}(\omega_i;y_i+\delta, 0).
\end{align*}
Then, the sampling algorithms are obtained as follows: 

\begin{itemize}
\item[-](Sampling from $\omega_i$) \ \ 
The full conditional distribution of $\omega_i$ is $PG(y_i+\delta, x_i^\top\beta+a_i+\xi(s_i)+\log\eta_i-\log\delta)$.

 \item[-](Sampling from $\beta$)\ \ 
 The full conditional distribution of $\beta$ is $N(\At_{\beta}, \Bt_{\beta})$, where 
\begin{align*}
\Bt_{\beta}&=\Big(\sum_{i=1}^{n}\omega_ix_ix_i^\top+B_{\beta}^{-1}\Big)^{-1},\
\At_{\beta}&=\Bt_{\beta}\sum_{i=1}^{n}x_{i}\Big\{\kappa_{i}-\omega_{i}(\log\eta_i+\log a_i+\xi_i-\log\delta)\Big\}.
\end{align*}

\item[-]
(Sampling of $\xi$) \ \ 
The full conditional distribution of $\mu_\xi$ is $N(\At_{\xi}, \Bt_{\xi})$, where 
\begin{align*}
\Bt_{\xi}&=\Big\{\sum_{i=1}^n\omega_{i}D_h(s_{i})D_h(s_{i})^\top + \tau H_h^{-1}\Big\}^{-1}, \\  
\At_{\xi}&=\Bt_{\xi}\sum_{i=1}^nD_h(s_i)\Big\{\kappa_{i}-\omega_{i}(\log\eta_i+\log a_i+x_i^\top\beta-\log\delta)\Big\}.
\end{align*} 
Then the posterior sample of $\xi$ is obtained as $\xi=D_h(s_{i})^{\top}\mu_\xi$.

\item[-]
(Sampling of $h$) \ \ 
The full conditional distribution of $h$ is proportional to 
\begin{align*}
&\pi(h)|H_h|^{-1/2}\exp\bigg(-\frac12\tau \mu_{\xi}^{\top}H_h^{-1}\mu_{\xi}
+\sum_{i=1}^n \kappa_iD_h(s_i)^\top\mu_{\xi}\\
& \ \ \ \ 
-\frac12\sum_{i=1}^n\omega_i\big\{x_i^\top\beta+D_h(s_i)^\top\mu_{\xi}+a_i+\log(\eta_i/\delta)\big\}^2\bigg),
\end{align*}
where $h\in (\underline{h}, \overline{h})$. 
We use the random-walk MH algorithm to generate $h$ from the distribution.

\item[-]
(Sampling of $\tau$)  \ \ 
The full conditional distribution of $\tau$ is ${\rm Ga}(A_{\tau}+n/2, B_{\tau}+ \mu_{\xi}^{\top}H_h^{-1}\mu_{\xi}/2)$.

\end{itemize}

\subsection{Posterior analysis with the scaled-beta distribution}
\label{supp:sb}

It is also possible to use the scaled-beta distribution for $\eta_{2i}$ in Equation~(8) (as distribution $H$) instead of the RSB distribution. The SB distribution has the same spike at zero, so the zero-inflation can be dealt with with this choice of error distribution. By contrast, the tail of the SB distribution is lighter than the RSB distribution. Thus, the SB and RSB distributions differ in their robustness to extremely large counts. 

The posterior analysis where $\eta_{2i} \sim \mathrm{SB}(a,b)$ is straightforward, as we see in this section. The density of $\mathrm{SB}(a,b)$ is, as provided in Section~1,
\begin{equation*}
\pi_{\rm SB}(\eta ;a,b) = \frac{1}{B(a,b)} \eta ^{a-1} (1+\eta )^{-(a+b)}, \ \ \ \ \ \eta >0.
\end{equation*}
It is well-known that the SB distribution is the marginal distribution of the following gamma-gamma mixture:
\begin{equation*}
    \eta | u \sim \mathrm{Ga}(a, u), \ \ \ \ \ \ u \sim \mathrm{Ga}(b,1).
\end{equation*}
Thus, augmenting $\eta _{2i}$ with latent variable $u_i$ simplifies the posterior computation. Most of the Gibbs sampler for the RSB model can be used ``as is,'' while we modify the sampling of $\eta_{2i}$ and $u_i$ as follows:
\begin{itemize}
    \item[-] 
(Sampling of $\eta_{2i}$) \ \  The full conditional distribution of $\eta_{2i}$ is ${\rm Ga}(y_i+a, \la_i+u_i)$ if $z_i=1$ and ${\rm Ga}(a, u_i)$ if $z_i=0$.

\item[-]
(Sampling of $u_i$) \ \ The full conditional distribution of $u_i$ is ${\rm Ga}(a+b, \eta_{2i} + 1)$. 
\end{itemize}

In the numerical examples of the main text, we used $a=b=1/2$.

\section{Additional simulation results}
To investigate the sensitivity of the choice of hyperparameters, $(a,b)$, in RSB, we applied RSB with all the combinations of $a,b\in \{1/4, 1/2, 3/4\}$ by using the same data generating process in Section~4.1.
To evaluate the performance, we computed the mean squared errors (MSE) of the point estimates of $\beta$, scaled mean squared errors (SMSE) of the Poisson intensity, and interval scores (IS) of the $95\%$ credible intervals of $\beta$, based on $R=500$ replications of the simulation.
The results are summarized in Figure~\ref{fig:sim-nst-supp}. 
It is shown that the differences caused by the choice of $(a,b)$ are minimal, particularly compared to the amount of improvement of RSB regression over the other methods as shown in Figure~2 in the main text.

\section{Predictive analysis under the RSB mixture model}

Although not discussed in the main text, predictive analysis could often be of the main interest in practice. Similarly to other statistical models used in applied studies, the predictive distribution obtained from the Poisson regression with the RSB mixture error has no closed form. To evaluate the predictive distribution, we suggest the standard, simulation-based procedure in Bayesian analysis, using the generated MCMC samples. 

For simplicity, we assume that the Poisson regression is used ($\gamma _i= 1$) with no random effect ($\xi_i=0$). 
Let $y_{1:n}$ be the observed data, and $y_{\ast}$ be the predicted value with fixed covariate $x_{\ast}$. Then, the predictive distribution of interest is the conditional distribution, 
\begin{equation*}
\begin{split}
    p(y_{\ast}|x_{\ast},y_{1:n}) &= \int \!\!\! \int \mathrm{Poisson}(y_{\ast}| \eta_{\ast} e^{ x_{\ast}^{\top}\beta } ) \ p(\eta _{\ast} , \beta | y_{1:n})  d\eta_{\ast} d\beta \\
    &= \int \!\!\! \int \!\!\! \int \mathrm{Poisson}(y_{\ast}| \eta_{\ast} e^{ x_{\ast}^{\top}\beta } ) \ \{ s \delta _1(\eta_{\ast}) + (1-s) \pi _{\mathrm{RSB}} (\eta _{\ast};a,b) \} \ p( \beta , s | y_{1:n})  d\eta_{\ast} d\beta ds,
\end{split}
\end{equation*}
where $p(\beta ,s |y_{1:n})$ is approximated by the MCMC samples. This integral representation helps constructing the MCMC samples from the predictive distribution of $y_{\ast}$ as follows. At each scan $t\in 1:T$ of the MCMC algorithm, given the generated sample $(\beta^{(t)}, s^{(t)})$, 
\begin{enumerate}
    \item Generate binary latent variable $z_{\ast}^{(t)}$ from $\mathrm{Bernoulli}(s^{(t)})$.

    \item If $z_{\ast}^{(t)} = 1$, then set $\eta _{\ast}^{(t)} = 1$. 

    If $z_{\ast}^{(t)} = 0$, then generate $\eta _{\ast}^{(t)} \sim \pi _{\mathrm{RSB}}(\eta _i;a,b)$. In doing so, utilize the hierarchical expression given in Section~2 and generate $w_{\ast}^{(t)} \sim \mathrm{Be}(a,1-a)$, $v_{\ast}^{(t)}\sim \mathrm{Ga}(a,w_{\ast}^{(t)})$, $u_{\ast}^{(t)}\sim \mathrm{Ga}(v_{\ast}^{(t)},1)$, and $\eta _{\ast}^{(t)}\sim \mathrm{Ex}(u_{\ast}^{(t)})$.

    \item Generate $y_{\ast}^{(t)} \sim \mathrm{Poisson}( \eta_{\ast}^{(t)} e^{ x_{\ast}^{\top}\beta ^{(t)} } )$.
\end{enumerate}
Then, the generated sequence, $y_{\ast}^{(1)},\dots, y_{\ast}^{(T)}$, approximates the target predictive distribution. For example, we can make point/interval predictions by computing the median and upper/lower quantiles of $y_{\ast}^{(1)},\dots, y_{\ast}^{(T)}$. 

External information can be integrated into this predictive analysis via conditioning. For example, plugging the point estimate of $\beta$ in the likelihood can be justified if one believes that the posterior of $\beta$ is sufficiently peaked with a large sample. 
Another example is that, if one believes that the future point, $y_{\ast}$, is not an outlier, then s/he should compute 
\begin{equation*}
    p(y_{\ast}|x_{\ast},y_{1:n},\eta_{\ast} = 1 ) = \int \mathrm{Poisson}(y_{\ast}| e^{ x_{\ast}^{\top}\beta } )  \ p( \beta  | y_{1:n})  d\beta. 
\end{equation*}
The resulting predictive distribution is expected to have lighter tails and finite moments.

%  Figure
\begin{figure}[!htb]
\centering
\includegraphics[width=12.5cm,clip]{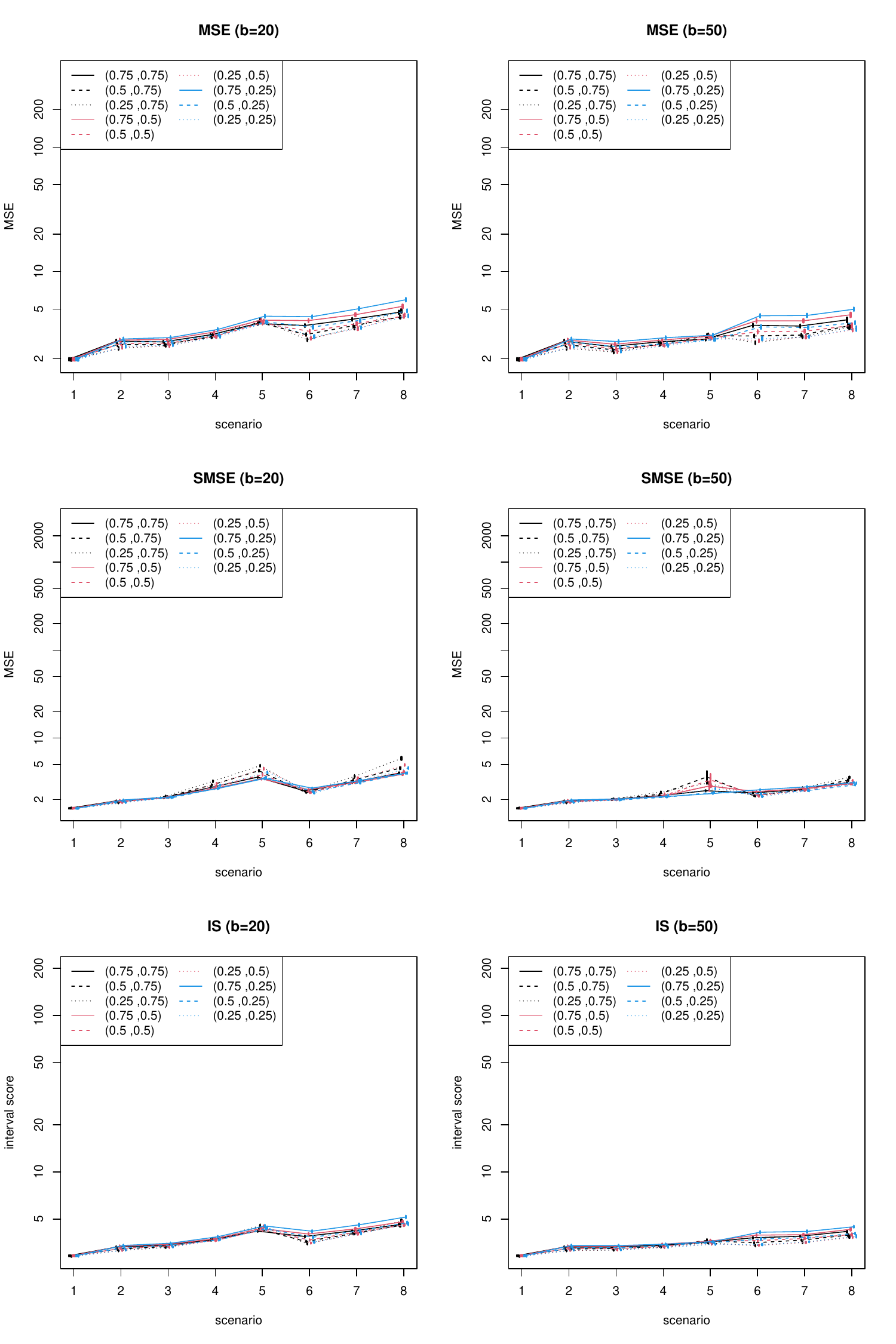}
\caption{MSE, SMSE and IS of RSB regression models with various choices of hyperparameters under 8 scenarios. 
The vertical lines are approximated $99\%$ confidence intervals of the MSE, SMSE and IS based on the estimated Monte Carlo errors.
The scales of the vertical axis are the same as those in Figure~2 in the main text.
\label{fig:sim-nst-supp}
}
\end{figure}

\end{document}